\documentclass[12pt,a4paper]{article}

\usepackage[margin=2.5cm]{geometry}
\usepackage{xcolor}

\usepackage{setspace}
\usepackage{amsmath,amssymb,amsfonts}
\usepackage{amsthm}

\newif\iffast
\fastfalse          

\iffast
  \usepackage[draft]{graphicx}   
\else
  \usepackage{graphicx}
\usepackage{float} 
\fi
\usepackage{pdfpages}
\usepackage{booktabs}
\usepackage{dcolumn}
\usepackage{float}        
\usepackage{caption}
\usepackage{subcaption}
\usepackage{multirow}
\usepackage[title]{appendix}
\usepackage{textcomp}
\usepackage{listings}
\usepackage{authblk}
\usepackage[hidelinks]{hyperref}
\hypersetup{
  colorlinks=true,
  linkcolor=blue,
  urlcolor=blue,
  citecolor=blue
}
\newcommand{\DOI}[1]{\href{https://doi.org/#1}{https://doi.org/#1}}

\newtheorem{lemma}{Lemma}

\begin{document}

\title{Heart Failure's First Shock and Nurse-Led Chronic Care}

\author{
Moslem Rashidi\thanks{Corresponding author.\ Department of Economics, University of Bologna, Piazza Scaravilli, 40126 Bologna, Italy (email: moslem.rashidi2@unibo.it).} \quad
Luke Brian Connelly\thanks{Centre for the Business and Economics of Health, The University of Queensland, St Lucia, QLD 4072, Australia (email: l.connelly@uq.edu.au), and Department of Sociology and Business Law, University of Bologna (email: luke.connelly@unibo.it).} \quad
Gianluca Fiorentini\thanks{Department of Economics, University of Bologna, Piazza Scaravilli, 40126 Bologna, Italy (email: gianluca.fiorentini@unibo.it).}
}

\date{\today}

\maketitle

\begin{abstract}

We study how a first heart-failure hospitalization—an adverse health shock—changes patients’ care, and whether a nurse-led chronic-care program sustains those post-shock investments. Using linked population-wide administrative records from Italy’s Romagna Local Health Authority (2017–2023), we anchor event time at each patient’s first CHF admission and exploit staggered timing to estimate dynamic effects. The shock triggers a sharp post-discharge surge: beta-blocker adherence, cardiology follow-up, and echocardiography rise immediately, while emergency-room use spikes just before admission and then stabilizes. We then estimate the incremental impact of enrollment in the Nurse-led Program for Chronic Patients (NPCP) using the interaction-weighted event-study estimator for staggered adoption. Under conventional diff-in-diffs inference, NPCP strengthens long-run preventive engagement, with little detectable change in emergency-room use. HonestDiD sensitivity analysis indicates these gains are economically meaningful but not statistically definitive under modest departures from parallel trends.
\end{abstract}

\begin{center}

\bigskip


\end{center}

\vfill

\thispagestyle{empty}


\section{Introduction}
\label{sec:Intro}

Major health events often lead individuals and families to reassess preventive behavior. After a severe family health shock, spouses and adult children increase their use of preventive care, with large, persistent responses (Fadlon and Nielsen, 2019). After a new diagnosis such as diabetes, households reduce calorie purchases, mainly by cutting unhealthy items (Oster, 2018). These patterns fit models in which a shock reveals information about disease risk, raises health salience, and induces a re-optimization of health investments (Grossman, 1972; Hoagland, 2025). Consistent with this mechanism, Hoagland (2025) shows that new chronic diagnoses can generate persistent increases in household healthcare use, including preventive care that remains elevated for years, while Darden (2017) finds that major chronic health shocks induce an immediate reduction in smoking, although his analysis does not establish whether that response later persists or attenuates.

This raises an important economic question: can structured post-shock care sustain or amplify the preventive responses triggered by a health event? Without support, patients and families may underinvest in prevention once the immediate shock recedes, accelerating disease progression and increasing costly episodes. In a Grossman-style health-capital framework, medication adherence, specialist follow-up, and diagnostic monitoring are inputs into future health. If patients face information, coordination, or attention frictions after discharge, these inputs may be chosen suboptimally. A nurse-led chronic-care program may reduce those frictions through education, information reinforcement, and active monitoring.

We study this question in congestive heart failure (CHF), a chronic and progressive condition with high morbidity and mortality. CHF patients often experience recurrent acute decompensations, and nearly one-fourth of hospitalized patients are readmitted within 30 days (Groenewegen et al., 2020). In Emilia-Romagna, Italy, the Nurse-led Program for Chronic Patients (NPCP) assigns specialist nurses to support discharge planning, education, self-monitoring, and proactive follow-up by phone or home visits. The program aims to reduce informational, behavioral, and coordination barriers that limit adherence, specialist care, and the early recognition of worsening symptoms. Prior evidence on heart-failure disease management suggests benefits for some outcomes, but evidence on readmissions is mixed. Large trials such as Tele-HF and BEAT-HF found no significant reduction in 180-day post-discharge readmission outcomes (Chaudhry et al., 2010; Ong et al., 2016). Consistent with this mixed pattern, the Cochrane review by Inglis et al. (2015) reports benefits for all-cause mortality and heart-failure-related hospitalizations, but not clear reductions in all-cause hospitalization.

Using linked administrative data from the Romagna Local Health Authority (2017--2023), we examine whether NPCP improves preventive engagement after a first CHF hospitalization. The data follow patients over time and record prescriptions, specialist visits, echocardiograms, and emergency room (ER) visits. Our design exploits two features. First, the first CHF hospitalization is a sharp, clinically meaningful shock with idiosyncratic timing. Second, patients enroll in NPCP at different dates relative to that shock and remain enrolled thereafter. We therefore define event time in calendar quarters around the first CHF admission and compare outcomes for enrolled patients with those of otherwise similar CHF patients who experience the same shock but do not enroll. This staggered-adoption difference-in-differences design allows us to trace the dynamic effects of NPCP on medication adherence, specialist follow-up, diagnostic monitoring, and ER use.

Because enrollment is rare in our CHF sample, selection into NPCP is central to interpretation. The institutional setting nonetheless supports a causal interpretation. Enrollment is initiated by the patient’s general practitioner (GP) and is shaped by supply-side constraints. Some GPs are more likely to refer eligible patients, and districts differ in their enrollment capacity. GP choice is not frictionless, since capacity limits make switching difficult, and rollout intensity varies across districts and over time. Consistent with this setting, we document differences in enrollment propensity across both GPs and districts. Thus, participation does not appear to reflect only patient motivation or baseline severity; access also depends on provider practice style and local capacity, which patients cannot easily modify in the short run.

Our contribution is both substantive and methodological. Substantively, we show how a major health shock and a nurse-led care model jointly shape the production of health investments. In a Grossman-style framework, medication adherence, physician visits, and diagnostic tests are inputs chosen after the shock, and we find that NPCP raises these inputs in the long run. We find no statistically detectable net effect on ER use, consistent with offsetting mechanisms: better disease control may reduce severe episodes, while closer monitoring and triage may increase precautionary ER visits. Methodologically, we implement modern event-study estimators that are robust to treatment-effect heterogeneity under staggered adoption (Sun and Abraham, 2021; Callaway and Sant’Anna, 2021).

The remainder of the article is organized as follows. Section \ref{sec:Background and Institutional Setting} describes the institutional context of NPCP. Section \ref{sec:Economic Framework} presents the economic framework. Section \ref{sec:Data and Identification} describes the data, identification strategy, and econometric approach. Section \ref{sec:Results} reports the main findings. Section \ref{sec:Conclusion} concludes.

\section{Background and Institutional Setting}
\label{sec:Background and Institutional Setting}

In the Romagna Local Health Authority (AUSL Romagna), Italy, the Nurse-led Program for Chronic Patients (NPCP) for chronic heart failure (CHF) was introduced in 2017 through collaboration among specially trained nurses, general practitioners (GPs), and specialists. The program is delivered mainly through in-person outpatient care within the existing primary-care network, either in Hub Community Health Centers (CHCs) managed by the local health authority or in Primary Care Groups (PCGs) run by groups of GPs. Its rollout was gradual: by 2022, only some hubs and structured PCGs had the nursing staff required to operate the program, while implementation in other districts occurred later.

Within each NPCP clinic, a chronic-care nurse serves as case manager. The nurse conducts an initial assessment, identifies patient needs, and develops a personalized care plan under GP supervision and in line with clinical guidelines. The nurse also provides health education, promotes self-care and healthy lifestyles, recalls patients for follow-up, and monitors key clinical indicators, including symptoms and weight, through scheduled contacts such as telephone calls. The GP remains the patient’s main treating physician and retains responsibility for evaluation, prescribing, medication adjustment, and referrals for tests or specialist consultations. When needed, a local cardiologist can be involved at the GP’s request. The program therefore creates a coordinated multidisciplinary pathway linking nurses, GPs, specialists, and other providers.

Enrollment is initiated by the GP based on explicit clinical criteria. Regional guidelines indicate referral for patients with a recent CHF hospitalization or diagnostic evidence of substantial cardiac dysfunction. Once enrolled, patients receive structured nurse-led follow-up, monitoring, and education. Patients not referred remain under usual GP and specialist care without dedicated nursing case management. Relative to standard care, NPCP’s distinguishing feature is proactive, coordinated, and continuous management rather than reactive, symptom-driven care.

\section{Economic Framework}
\label{sec:Economic Framework}
This section develops a parsimonious dynamic disease-management framework for three patterns in our setting: how a first CHF hospitalization changes post-discharge management, why initial improvements may fade, and how enrollment in the Nurse-led Program for Chronic Patients (NPCP) can sustain health investments and affect utilization beyond acute care. We anchor event time at the first CHF hospitalization and treat post-discharge management as a repeated investment problem aimed at future clinical stability.

In a Grossman-style model, health is a durable stock that yields utility and productive time, depreciates over time, and can be increased through investment (Grossman, 1972). In our setting, the relevant investments observed in administrative data are medication adherence, timely outpatient follow-up, and diagnostic monitoring. These actions require time, effort, and coordination, but improve expected future health and reduce the risk of acute deterioration. A first CHF hospitalization is a large adverse shock that lowers health and raises the return to stabilization. Thus, even without NPCP, patients may rationally increase management effort immediately after discharge. This is consistent with evidence that major health shocks increase preventive investments among affected households (Fadlon and Nielsen, 2019). Throughout, we use ``salience'' only as a shorthand for time-varying incentives and effective costs around discharge, not as a departure from rational choice.

ER visits are not preventive inputs but downstream outcomes reflecting acute episodes and triage decisions given evolving clinical stability. Better preventive management can reduce later acute events, but ER use need not decline monotonically because monitoring intensity, access to guidance, and escalation thresholds also shape whether symptoms lead to ER care.

The central challenge is sustaining investment over time. In a standard dynamic framework, the shadow price of health investment includes information frictions, time constraints, effort, and coordination burdens (Grossman, 1972). For CHF patients, these barriers may be lower just after hospitalization, when perceived risk and clinical attention are high, but rise as symptoms stabilize and routine contact declines. Preventive effort may therefore partially revert toward baseline. This is consistent with evidence that adherence to secondary-prevention therapies declines after acute cardiovascular events (Bahit et al., 2023).

NPCP is designed to reduce these frictions. As discussed in Section \ref{sec:Background and Institutional Setting}, it assigns a chronic-care nurse who provides case management, education, proactive follow-up, monitoring, and coordination with the GP and, when needed, with a territorial cardiologist. In the model, NPCP lowers the effective cost of sustained investment and can raise its productivity through earlier responses and better-organized care. We therefore expect higher adherence, more follow-up, and more monitoring after enrollment, with larger treated--control gaps in later periods as preventive effort among untreated patients attenuates. By contrast, the effect on ER use is ambiguous ex ante: better disease control may reduce deterioration-driven demand, while closer monitoring and triage may increase precautionary ER visits (Taubman et al., 2014; Chaudhry et al., 2010; Ong et al., 2016). Our empirical design tests these predictions around the first CHF hospitalization (see Appendix A (\ref{app:structural_model_short}).


\section{Data and Identification}
\label{sec:Data and Identification}

We use population-wide administrative health records from the Romagna Local Health Authority (LHA) covering all reimbursed healthcare utilization and outcomes from 2017 through 2023. Using a stable anonymized patient identifier, we link individual-level information on hospital admissions, outpatient specialist care, pharmaceutical dispensing, diagnostic testing, and the mortality registry. This linkage provides near-complete longitudinal follow-up within the LHA catchment area: attrition is negligible, with only 96 patients (0.54\% of 17,835) exiting before the end of 2023.

We initially identify 17,835 individuals diagnosed with congestive heart failure (CHF) during 2017--2023. For causal analysis, we restrict attention to patients who experience at least one acute hospitalization for CHF decompensation, yielding a final analytic sample of 15,703 patients; the remaining 2,132 never experience a first CHF hospitalization during the observation window. This restriction strengthens internal validity for three reasons. First, it aligns treated and untreated patients at a common, clinically meaningful severity threshold: an acute decompensation requiring hospitalization. This avoids comparing NPCP enrollees with milder CHF cases that never reach hospital care. Second, it reduces unmeasured baseline heterogeneity that observable risk adjusters, including the Multidimensional Comorbidity Score (MCDS) (Iommi et al., 2020), may not fully capture. Patients who never require hospitalization may differ in unobserved dimensions such as ventricular function, caregiver support, or diagnostic intensity that affect both enrollment and outcomes. Third, it provides a sharp and common clinical benchmark---the first CHF hospitalization---that anchors event time in our event-study design.

We treat the first hospitalization for CHF decompensation as a plausibly exogenous health shock in the sense that its timing is quasi-random within a short window. Acute decompensation admissions are typically triggered by abrupt deterioration and an emergency presentation, making their exact timing difficult for patients to anticipate or manipulate. To strengthen this interpretation, we define the benchmark as an index CHF admission preceded by a three-year clean period with no prior CHF decompensation hospitalizations. This reduces concern that the benchmark is mechanically part of an ongoing recent sequence of acute episodes and limits the possibility that outcomes are already shifting because of a continuing decompensation cycle. Related work uses closely related timing-based strategies: Dobkin et al. (2018) study sudden hospital admissions, while Fadlon and Nielsen (2019, 2021) construct counterfactuals from households that experience the same severe shock later. In our setting, the identifying assumption is that, conditional on baseline health risk and local healthcare supply conditions, the quarter of the index admission is quasi-random when comparing adjacent quarters within the event-study window.

The final analytic sample therefore comprises 15,703 CHF patients. Of these, 494 enroll in the Nurse-led Program for Chronic Patients (NPCP) by the end of the study period, while 15,209 never enroll and serve as untreated controls.

Because enrollment is rare in our CHF cohort (about 3\%), selection into NPCP is central to interpretation. We therefore provide institutional and empirical evidence that participation is shaped importantly by supply-side constraints rather than by patient demand alone. First, patients do not have unconstrained choice of general practitioner (GP). They can select only among local GPs with available capacity, and GP panels are typically capped at around 1,500 patients, with limited scope for temporary or exceptional increases (SISAC, 2018; Conferenza Stato-Regioni, 2024). During periods or in areas with GP shortages, the effective GP choice set is therefore narrow and switching may be costly or delayed. Consistent with this, regional saturation is substantial: according to a summary of SISAC-based figures reported by the provincial medical board, 57.6\% of Emilia-Romagna GPs exceed the 1,500-patient threshold and the estimated regional shortfall is 536 GPs, implying that closed or full lists are common in shortage areas (Fondazione GIMBE, 2025). Second, enrollment flows vary sharply across districts and over time. In several districts, enrollment remains close to zero for long stretches and then rises discretely, a pattern more consistent with delayed rollout or step changes in local implementation capacity than with gradual changes in patient health literacy or social networks (Figure~\ref{Fig:scatter_data}). Third, we document substantial GP-level heterogeneity in enrollment propensity. Using a leave-one-out measure of each GP’s tendency to enroll eligible patients (Goldsmith-Pinkham et al., 2025; Chyn et al., 2025), we find that treated patients are associated with GPs whose enrollment propensities are much higher on average (14.9\%) than those associated with control patients (2.7\%; $p<0.001$) (Table~\ref{tab:gp_leniency_loo}). Taken together, these patterns indicate that access depends importantly on district-level implementation capacity and GP practice style, not solely on patient selection.

Supply-side constraints, however, do not by themselves rule out baseline differences between enrollees and non-enrollees. We therefore report balance in Table~\ref{tab:Sum2b} and use those comparisons to interpret the pre-period level gaps visible in Figure~\ref{fig:scatter1}. Enrolled patients are younger and more often male, while baseline clinical complexity is only modestly lower among enrollees. By contrast, deprivation is essentially identical across groups. These differences are economically and clinically meaningful: younger and less complex CHF patients may face fewer contraindications, fewer treatment interruptions, and lower organizational burdens, making sustained chronic-care management easier.

Our event-study is anchored at each patient’s first CHF hospitalization ($\tau = 0$), while outcomes are measured quarterly and hospitalization and NPCP enrollment may occur in different quarters. This creates transition quarters between the two events. For patients who enroll after hospitalization, the quarters from the hospitalization quarter through the quarter before enrollment are post-shock but pre-program, so outcomes in that interval reflect the acute hospitalization response, discharge planning, and recovery dynamics rather than exposure to NPCP. For patients who enroll before hospitalization, the quarters from the enrollment quarter through the hospitalization quarter are already post-enrollment but pre-shock, so retaining them would place treated observations into the pre-hospitalization comparison period. The intuition for gap censoring is therefore to avoid assigning these transition quarters to either the untreated or treated state. If we failed to censor them, the estimated NPCP effect would be contaminated either by shock-driven changes unrelated to the program or by already-treated observations entering the baseline period. This would blur the distinction between the common hospitalization response and the incremental effect of NPCP, potentially generating spurious short-run effects or artificial pre-trends. For this reason, we deterministically set these transition-gap quarters to missing based only on observed timing, so that identification comes only from quarters that are clearly pre-enrollment or post-enrollment.

We study outcomes chosen ex ante because they are measured symmetrically around the index hospitalization and map directly to the investment-versus-acute-use margins emphasized by a Grossman-style framework (Grossman, 1972) and CHF management guidelines (Takeda et al., 2012). Medication adherence is an indicator equal to one if the patient’s quarterly proportion of days covered (PDC) for $\beta$-blockers (ATC code C07) exceeds 75\%. Cardiology visits, echocardiograms, and emergency room (ER) visits are each defined as indicators for at least one occurrence in the quarter (Table~\ref{tab:Sum1a}).

\begin{table}[htbp]
\centering
\captionsetup{skip=10pt} 
\caption{Summary statistics for quarterly outcomes, including the number of patient--quarter observations, unique patients, means, within- and between-patient standard deviations, overall standard deviations, and ranges.}
\label{tab:Sum1a}
\resizebox{\textwidth}{!}{%
    
\begin{tabular}{lrrrrrrr}
\hline
\textbf{Variable} & \textbf{Total Obs.} & \textbf{Unique Patients} & \textbf{Mean} & \textbf{SD (Within)} & \textbf{SD (Between)} & \textbf{Total SD} & \textbf{Range} \\
\hline
Adherence to beta-blockers (PDC $\geq$ 75\%) & 187,513 & 15,665 & 0.35 & 0.23 & 0.42 & 0.48 & [0, 1] \\
Cardiology visits                  & 187,513 & 15,665 & 0.11 & 0.27 & 0.17 & 0.32 & [0, 1] \\
Echocardiogram performed                & 187,513 & 15,665 & 0.04 & 0.18 & 0.08 & 0.20 & [0, 1] \\
ER visits                          & 187,513 & 15,665 & 0.17 & 0.34 & 0.14 & 0.37 & [0, 1] \\
\hline
\end{tabular}
}
\begin{flushleft}
\footnotesize \textit{Notes:} These outcomes are defined at the patient-quarter level and observed before the first CHF hospitalization (event time $ \tau = 0 $).
\end{flushleft}
\end{table}

To visualize dynamics and assess identification, Figure~\ref{fig:scatter1} plots moving-average trajectories from 12 quarters before to 12 quarters after the first CHF hospitalization, with event time normalized to the index hospitalization quarter ($\tau = 0$). In the clean pre-period ($\tau \le -2$), treated and untreated patients follow closely parallel paths in adherence, cardiology visits, echocardiograms, and ER use, supporting the parallel-trends assumption. At $\tau = -1$, ER admissions rise sharply for both groups, consistent with shared acute deterioration and pre-admission processing before the index hospitalization. We therefore treat $\tau = -1$ as a transitional quarter rather than a clean pre-period, since utilization may mechanically reflect deterioration-to-admission workflows (including bundled pre-admission work-up) and, for patients who enroll before the index admission, partial NPCP exposure. Accordingly, we use $\tau = -2$ as the omitted/reference period and assess pre-trends over $\tau \le -2$. A related concern is that movements around $\tau = -1$ may reflect selection into NPCP based on initial-shock severity. Two features of the data make this unlikely. First, 121 of the 494 eventual enrollees enter NPCP before the first hospitalization, so fluctuations at $\tau = -1$ cannot reflect anticipatory responses to the index-hospitalization shock for this subgroup. Second, among patients who enroll only after the first hospitalization ($N = 373$), the moving-average trajectories show no pronounced differential spike at $\tau = -1$ in cardiology visits or echocardiograms relative to controls (Figure~\ref{Fig:movAveexcl121}), and enrollment is only occasionally contemporaneous with the index hospitalization; it usually occurs in later quarters, often after multi-quarter lags. This timing, together with the absence of a sharp differential spike at $\tau = -1$, is inconsistent with a mechanism in which pre-index shock intensity drives immediate enrollment; instead, enrollment is more consistent with the supply-side implementation and access constraints documented above.

\begin{figure}[htbp]
  \centering
  \includegraphics[width=\textwidth]{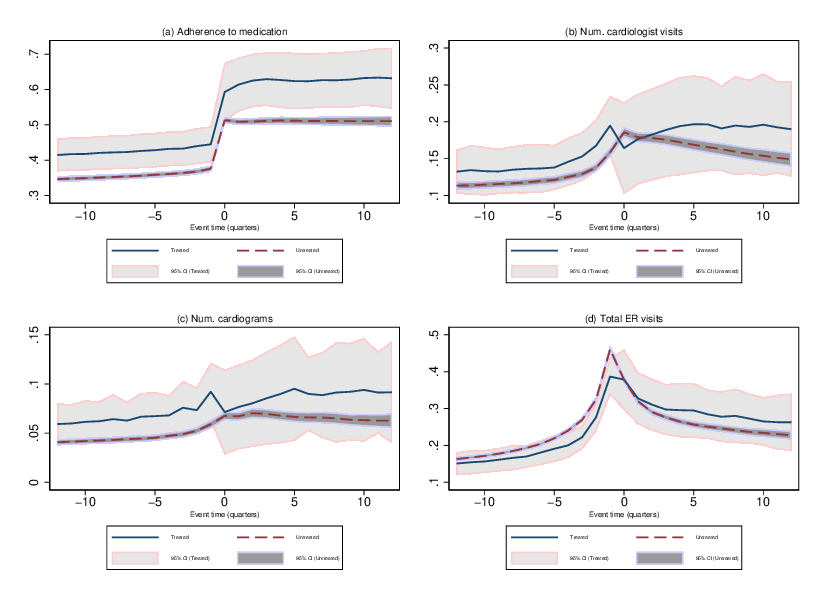}
\caption{Moving-average event-time trajectories around the first CHF hospitalization (\(t=0\)) for panels (a)--(d): \(\beta\)-blocker adherence, cardiologist visits, echocardiograms, and ER visits. Shaded bands show 95\% CIs over \(\pm 12\) quarters.}
  \label{fig:scatter1}
\end{figure}

Appendix~C~\ref{App:DES_Graphs} reports stratified trajectories by age, gender, and baseline clinical complexity (MCDS). These are clinically and economically salient sources of heterogeneity in vulnerability and in the capacity to adjust preventive management after a major health shock. Because treated and untreated patients are nearly identical in the deprivation index at baseline (difference of less than $0.01$ SD; Table~\ref{tab:Sum2b}), further deprivation stratification adds little interpretive value and would reduce precision.

\begin{table}[htbp]
\centering
\caption{Baseline characteristics and pre-hospitalization balance. Panel A: overall means. Panel B: by NPCP enrollment with SMD. Deprivation: 1=least, 5=most. MCDS: Multisource Comorbidity Score (higher = worse).}
\label{tab:Sum2}

\begin{subtable}[t]{\textwidth}
\centering
\caption*{Panel A. Overall baseline distributions}
\resizebox{\textwidth}{!}{%
\begin{tabular}{lrrrrr}
\toprule
\textbf{Variable} & \textbf{Total Obs.} & \textbf{Unique Patients} & \textbf{Mean} & \textbf{SD (Between)} & \textbf{Range} \\
\midrule
Sex (Female = 1)         & 188{,}436 & 15{,}703 & 0.52 & 0.50 & [0, 1] \\
Age (years)              & 188{,}436 & 15{,}703 & 79.57 & 10.84 & [16, 112] \\
Comorbidity Index (MCDS) & 188{,}436 & 15{,}703 & 4.02 & 1.58 & [1, 6] \\
Deprivation Index (1--5) & 177{,}741 & 14{,}814 & 2.90 & 1.45 & [1, 5] \\
\bottomrule
\end{tabular}
}
\end{subtable}

\vspace{0.6em}

\begin{subtable}[t]{\textwidth}
\centering
\caption*{Panel B. Baseline balance by NPCP enrollment (Treated vs. Control)}
\resizebox{\textwidth}{!}{%
\begin{tabular}{lrrrrrrr}
\toprule
\textbf{Variable} & \textbf{N (Control)} & \textbf{Mean (Control)} & \textbf{N (Treated)} & \textbf{Mean (Treated)} & \textbf{Diff. (T--C)} & \textbf{t-stat} & \textbf{SMD} \\
\midrule
Sex (Female = 1)         & 182{,}508 & 0.53 & 5{,}928 & 0.41 & 0.12 & 18.74 & 0.24\textsuperscript{**} \\
Age (years)              & 182{,}508 & 79.68 & 5{,}928 & 76.23 & 3.45 & 25.77 & 0.32\textsuperscript{**} \\
Comorbidity Index (MCDS) & 182{,}508 & 4.03 & 5{,}928 & 3.94 & 0.09 & 4.46  & 0.06 \\
Deprivation Index (1--5) & 172{,}008 & 2.90 & 5{,}733 & 2.88 & 0.02 & 1.44  & 0.01 \\
\bottomrule
\end{tabular}
}
\end{subtable}

\begin{flushleft}
\footnotesize \textit{Notes:} Differences are Treated minus Control. SMD = standardized mean difference. Variables with \textsuperscript{**} show notable imbalance and are controlled for in the analysis.
\end{flushleft}
\label{tab:Sum2b}
\end{table}

\clearpage
\subsection{Empirical Design}
\label{sec:Empirical Design}

Identification is anchored on each patient’s first acute CHF hospitalization. Let $\tau_{it}=0$ denote patient $i$’s hospitalization quarter in calendar time $t$, with $\tau_{it}<0$ and $\tau_{it}>0$ indicating quarters before and after hospitalization. Treatment is defined at the patient--quarter level as enrollment in NPCP:
$D_{it}=\mathbf{1}\{\text{patient } i \text{ is enrolled in quarter } t\}$.
Enrollment is absorbing, so once $D_{it}=1$, it remains equal to one in all subsequent quarters.

We study four quarterly outcomes $Y_{it}$ constructed consistently from the administrative records. Three capture preventive disease-management behavior: $\beta$-blocker adherence (PDC $\geq 75\%$), specialist cardiology visits, and echocardiography. The fourth is ER use, which captures acute-care utilization and complements the preventive outcomes by measuring short-run clinical deterioration and high-intensity care use. All four outcomes are observed symmetrically before and after the index hospitalization.

Because index hospitalizations occur in different calendar quarters and NPCP availability changes over time, treatment timing is staggered across cohorts. In this setting, standard two-way fixed-effects (TWFE) event-study regressions can be misleading. With heterogeneous treatment effects, TWFE can aggregate cohort-specific effects using non-convex and sometimes negative implicit weights (Goodman-Bacon 2021; de Chaisemartin and D’Haultfœuille 2020; Sun and Abraham 2021). As a result, the estimated dynamic profile can be distorted even when all cohort-specific effects have the same sign. This issue is especially relevant here because common calendar-time shocks, most notably the COVID-19 pandemic, plausibly altered the post-hospitalization environment differently across cohorts.

Our primary estimator is the interaction-weighted event-study estimator of Sun and Abraham (2021), implemented via \texttt{xtevent} (Freyaldenhoven et al., 2021, 2025). We use this approach because it is well suited to staggered adoption with heterogeneous dynamic treatment effects, avoids the problematic implicit weighting of conventional TWFE event studies, and fits naturally within a linear fixed-effects panel framework. Conceptually, it recovers cohort-specific relative-time effects (CATTs) and aggregates them with interaction weights, so the resulting event-time coefficients remain interpretable even when treatment effects vary across cohorts and over exposure time. As a robustness check, we also estimate the Callaway and Sant'Anna (2021) framework, which identifies group-time average treatment effects, ATT$(g,t)$, under limited anticipation and conditional parallel trends, and then aggregates them into event-time summaries using explicit weights. The two approaches therefore target closely related untreated-potential-outcome effects, but differ in how cohort-time effects are estimated and aggregated. They may also rely on different comparison groups: Callaway and Sant’Anna allow either never-treated or not-yet-treated controls, whereas the baseline Sun–Abraham implementation uses a fixed control cohort (never-treated units when available, or the last-treated cohort otherwise), although Sun and Abraham also note that alternative control sets can be used. Accordingly, similar results across the two methods are especially informative when the implementations differ in comparison group and/or covariate adjustment.

For each outcome $Y_{it}$, we estimate the following fixed-effects event-study regression following Freyaldenhoven et al. (2021):
\begin{equation}
\label{eq:es}
\begin{aligned}
Y_{it}
&= \sum_{\tau=-G-L_G}^{M+L_M-1} \delta_{\tau}\,\Delta D_{i,t-\tau}
\;+\; \delta_{M+L_M}\,D_{i,t-M-L_M}
\;+\; \delta_{-G-L_G-1}\big(1 - D_{i,t+G+L_G}\big) \\
&\quad+\; \alpha_i \;+\; \lambda_t \;+\; \mu_{r(i),t} \;+\; \phi_{g(i,t)}
\;+\; X_{it}'\theta \;+\; \varepsilon_{it}.
\end{aligned}
\end{equation}

The lead and lag coefficients $\delta_{\tau}$ capture dynamic differences by event time, and we normalize $\delta_{-2}=0$ because $\tau=-1$ is a transitional quarter. The indicators outside the main window absorb event times beyond the chosen lead--lag horizon, allowing observations far from the event to remain in the estimation sample without imposing extrapolation. In practice, we choose the event window to balance clinical relevance against data support at each horizon. The specification includes rich controls for confounding. Patient fixed effects $\alpha_i$ absorb time-invariant heterogeneity, including persistent differences in baseline severity and comorbidity. Calendar-quarter effects $\lambda_t$ capture aggregate shocks, while district-by-quarter effects $\mu_{r(i),t}$ flexibly absorb local time-varying supply conditions, policy changes, and differential pandemic intensity. We also include GP fixed effects $\phi_{g(i,t)}$ to account for provider-level practice style while allowing a patient’s primary physician to change over time. Time-varying covariates $X_{it}$ capture observable within-patient changes that may affect outcomes. This specification follows Freyaldenhoven et al. (2021) and, by construction, yields well-defined event-study contrasts under heterogeneous treatment effects.

Let $H_i$ denote patient $i$’s hospitalization quarter and $E_i$ denote the NPCP enrollment quarter. To ensure that pre-treatment event times are clearly untreated and post-treatment event times are clearly treated, we drop quarters with ambiguous treatment status. If $E_i>H_i$, we omit quarters $\tau\in\{0,\dots,E_i-H_i-1\}$, which occur after hospitalization but before enrollment. If $E_i<H_i$, we omit quarters $\tau\in\{E_i-H_i,\dots,-1\}$, which occur after enrollment but before hospitalization. Because these exclusions are deterministic functions of observed dates, they prevent contamination of pre-trends and avoid attributing shock-driven post-hospitalization utilization to NPCP before actual exposure (see Appendix~D (\ref{app:ignorable-censoring})).

Finally, we assess robustness to parallel-trends violations using the HonestDiD sensitivity analysis of Rambachan and Roth (2023). Rather than imposing exact parallel trends, HonestDiD allows the treated group’s untreated counterfactual to drift after treatment, subject to bounds tied to pre-treatment trend differences. We use the relative-magnitude restriction, in which the sensitivity parameter $M$ limits post-treatment deviations relative to the largest pre-treatment departure: $M=0$ imposes exact parallel trends, $M=1$ allows violations as large as the largest observed pre-trend deviation, and larger values allow proportionally larger departures. We report robust confidence intervals for the average post-hospitalization effect for $M\in\{0,0.5,1,1.5,2\}$. If the HonestDiD interval includes zero at a given $M$, the baseline positive estimate is not robust to a violation of that magnitude. This shows both how much our conclusions depend on parallel trends and how large a deviation would be needed to overturn the baseline inference.


\section{Results}
\label{sec:Results}

This section reports event-study evidence on outcomes around the first CHF hospitalization and on the incremental effect of NPCP enrollment. We first document the dynamics generated by the hospitalization shock itself, identified from staggered hospitalization timing. We then show that NPCP primarily strengthens longer-run preventive and follow-up care---medication adherence, cardiology visits, and echocardiograms---with little evidence of a corresponding change in ER utilization. Because program take-up is low, short-run estimates are less precise and should be interpreted with caution.

Our event-study figures are produced with \texttt{xtevent}/\texttt{xteventplot}, following the visualization recommendations in Freyaldenhoven et al. (2021, 2025). Each plot shows the estimated event-time path relative to a normalized reference period. Points indicate estimated effects, vertical bars indicate pointwise 95\% confidence intervals, and, when included, the outer band gives simultaneous (sup-t) inference. The x-axis is event time with binned endpoints, and the omitted reference period is normalized to zero. By default, the plots also display the mean of the outcome at the normalized period and $p$-values for Wald tests of no pre-trends and of whether post-event dynamics level off.

Before turning to NPCP, we first characterize how outcomes evolve around the first CHF hospitalization, which is the key clinical shock in our design. Because all patients in the analytic sample eventually experience a first hospitalization, identification comes from staggered timing rather than from comparison with a permanently never-shocked group. Following Sun and Abraham (2021), we use the latest hospitalization cohort as the reference cohort and its pre-hospitalization outcomes to construct the counterfactual path for earlier cohorts. Figures~\ref{fig:sa2T} and~\ref{fig:sa2UT} report these dynamics separately for patients who eventually enroll in NPCP and for those who do not. The estimated profiles show a strong post-shock increase in preventive and follow-up care. Medication adherence rises sharply at hospitalization, then attenuates over subsequent quarters and continues to evolve in the late post period, indicating no clear stabilization within the event window. This pattern is consistent with an immediate post-discharge strengthening of chronic management that gradually fades. Cardiology visits and echocardiograms also rise markedly in the hospitalization quarter and early follow-up, then decline as intensive monitoring tapers. By contrast, ER visits spike in the quarter immediately preceding hospitalization, then fall and stabilize afterward, consistent with acute deterioration in the run-up to admission and a smaller but persistent elevation after the shock (Table~\ref{tab:EST2}).

Building on these shock dynamics, we next estimate the incremental effect of NPCP enrollment using the Sun and Abraham (2021) interaction-weighted event-study estimator. The results, reported in Table~\ref{tab:EST2} and Figure~\ref{fig:sa_main}, show no evidence of differential pre-trends for any of the four outcomes. In the post-hospitalization period, the treated group exhibits a sustained strengthening of preventive care rather than a reduction in acute utilization. Medication adherence rises after hospitalization and remains higher in the long run, with a positive and statistically significant effect in the pooled $6+$ quarters bin. Cardiology visits and echocardiograms follow the same broad pattern: effects are imprecisely estimated in the first few post-enrollment quarters but become positive, large, and statistically significant from $6+$ quarters onward. This is consistent with NPCP improving care coordination and sustaining guideline-concordant follow-up after the immediate discharge phase. By contrast, ER visits remain close to zero throughout the post-hospitalization window and show no statistically meaningful long-run change.

Because only a small share of CHF patients enroll in NPCP, our estimates---especially at short horizons---are necessarily noisy. To clarify how much can be learned from statistically insignificant short-run coefficients, Table~\ref{tab:mde_avgpost} reports minimum detectable effects (MDEs) for the average post-treatment effect over event times 0--5 in the Sun--Abraham specification. These MDEs show that, over the first six post-hospitalization quarters, the design is mainly powered to detect moderate-to-large effects. Accordingly, insignificant short-run estimates rule out only large immediate impacts; they remain consistent with smaller effects below the design’s detectable threshold. This power limitation is even more relevant in subgroup analyses by age, gender, and baseline comorbidity, where smaller samples widen confidence intervals further. We therefore emphasize the size, timing, and persistence of the estimated effects, especially at longer horizons, rather than relying only on conventional significance thresholds.

We next examine heterogeneity by age, gender, and baseline comorbidity to better understand the program’s mechanisms. The identifying assumptions appear most credible for older patients (age $\geq 75$) and for women, for whom the pre-trend tests are supportive of a causal interpretation. In these groups, NPCP enrollment is associated with significantly higher preventive and follow-up care---including several percentage-point increases in $\beta$-blocker adherence, cardiology visits, and echocardiograms (all $p<0.05$)---with no meaningful change in ER use. By contrast, younger patients (age $<75$) and men display statistically significant pre-trend violations for multiple outcomes, so those estimates are best read as descriptive rather than causal. For baseline comorbidity, both high- and low-MCDS patients show sizable long-run gains in planned care, particularly cardiology and echocardiographic follow-up, with larger magnitudes among the higher-risk group. ER patterns differ across comorbidity strata---a small increase for high MCDS and a decline for low MCDS---but in both cases pre-trend tests are rejected, so those ER differences should be interpreted cautiously. Overall, the subgroup results that are most credible point to a common mechanism: NPCP mainly strengthens ongoing chronic-care management for higher-risk patients rather than directly reducing acute ER utilization.

We also assess robustness using the uniform confidence bands of Rambachan and Roth (2023). Figures~\ref{fig:honestdid}, \ref{fig:honestdid_lt75}, \ref{fig:honestdid_ge75}, \ref{fig:honestdid_female}, \ref{fig:honestdid_male}, \ref{fig:honestdid_MCDS_poor}, and \ref{fig:honestdid_MCDS_good} show that the HonestDiD 95\% robust confidence intervals include zero for all outcomes across sensitivity values $M \in [0,2]$, including the benchmark case $M=0$. At the same time, the underlying point estimates remain positive for medication adherence and specialist/diagnostic use, and remain near zero for ER visits, which is consistent with the qualitative pattern implied by NPCP. The wide HonestDiD intervals, however, imply that under conservative inference these improvements are not statistically distinguishable from zero. We therefore interpret the long-run positive patterns as suggestive rather than definitive evidence of benefit.

\begin{figure}[htbp]
  \centering
  \includegraphics[width=.82\linewidth,trim=0pt 5pt 0pt 0pt, clip]{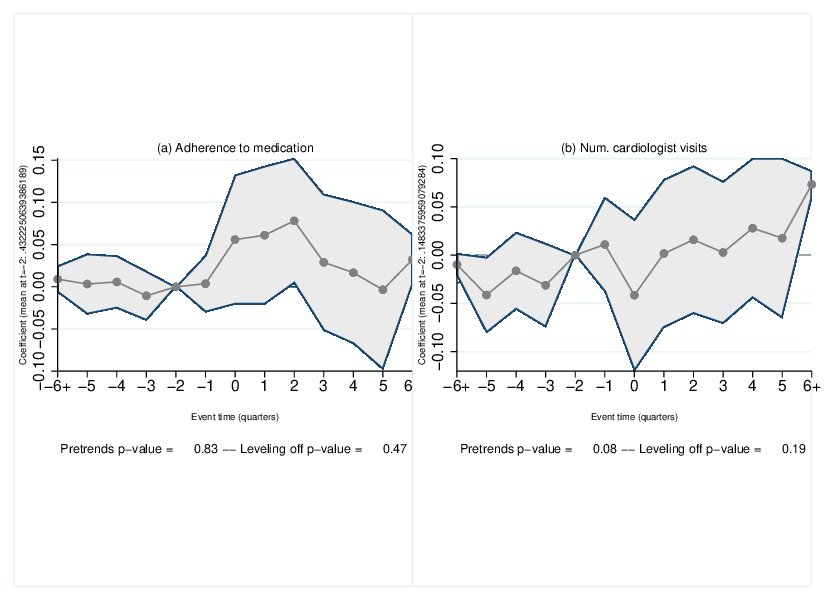}

  \vspace{0.4em}

  \includegraphics[width=.82\linewidth,trim=0pt 0pt 0pt 0pt, clip]{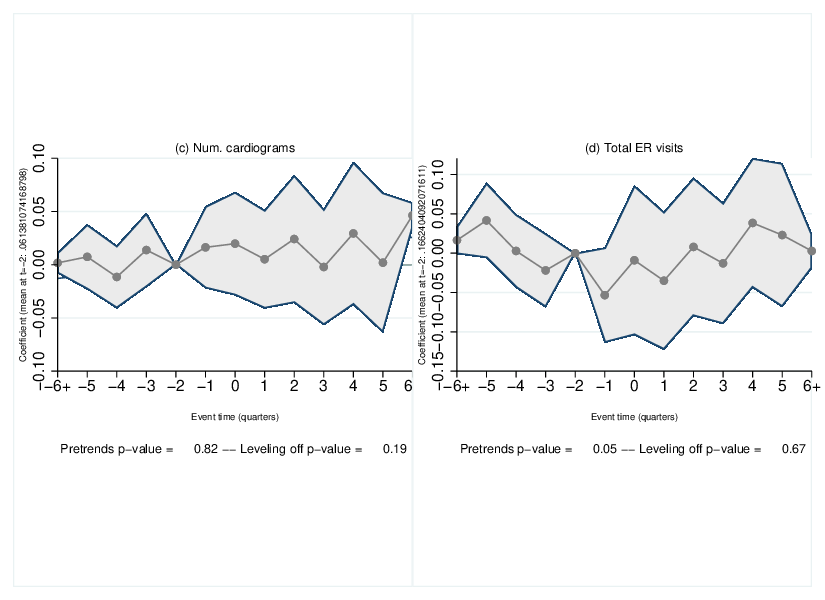}

\caption{NPCP effects from the Sun--Abraham (2021) event study. Event time is in quarters around the first CHF hospitalization (\(\tau=0\); reference = \(-2\)). Panels report \(\beta\)-blocker adherence, cardiology visits, echocardiograms, and ER visits. Points are estimates; bars are 95\% CIs; the long-run bin pools \(\tau \ge 6\).}
  \label{fig:sa_main}
\end{figure}

\begin{figure}[htbp]
    \centering
    \includegraphics[width=\textwidth]{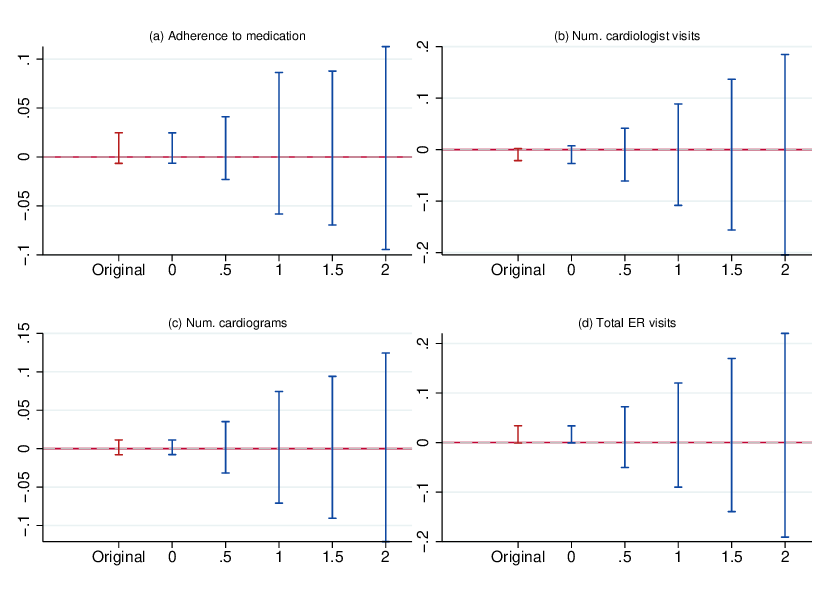}
\caption{HonestDiD sensitivity analysis (Rambachan and Roth, 2023) for the Sun--Abraham estimates. The figure reports 95\% confidence intervals across \(M \in \{0,0.5,1,1.5,2\}\), where larger \(M\) permits larger departures from parallel trends.}
    \label{fig:honestdid}
\end{figure}

As a robustness check, we re-estimate the dynamic effects using the Callaway and Sant’Anna (2021) difference-in-differences framework with inverse probability weighting (Figure~\ref{fig:cs_main} and Table~\ref{tab:pre-trend_CS}). The resulting event-time profiles closely resemble the baseline Sun--Abraham estimates, although short-run coefficients are somewhat smaller and less precise. We also estimate a conventional homogeneous-effects event study (Figure~\ref{fig:ES_homo}); its dynamics are very similar to the baseline profiles, suggesting that the main patterns are not driven by extreme treatment-effect heterogeneity across cohorts.

Figures~\ref{fig:sa4} and \ref{fig:sa7} vary the definition of the pooled long-run bin. In the baseline, late post-periods are pooled to improve precision; here we instead define the pooled long-run category as \(\tau \ge 5\) (Figure~\ref{fig:sa4}) and \(\tau \ge 7\) (Figure~\ref{fig:sa7}), holding all other choices fixed. The estimated profiles are qualitatively unchanged. For \(\beta\)-blocker adherence and planned specialist/diagnostic care, effects remain small and imprecisely estimated in the early post-hospitalization quarters but turn positive in the pooled long-run bin under both alternatives, reinforcing the conclusion that NPCP strengthens sustained follow-up rather than generating immediate changes. ER visits remain close to zero throughout the post window under all pooling choices, indicating no systematic long-run shift in acute utilization.

Finally, to address the concern that patients already enrolled in NPCP before the first CHF hospitalization may blur the interpretation of post-discharge management, we re-estimate the main Sun and Abraham interaction-weighted event study after excluding the 121 pre-hospitalization enrollees and restricting the treated group to patients who enroll only after the index hospitalization (\(N=373\)). The results are qualitatively unchanged. Pre-hospitalization coefficients remain close to zero, with no evidence of differential pre-trends for adherence and ER use (pre-trend $p$-values 0.76 and 0.83), weak evidence against parallel trends for cardiology visits ($p=0.10$), and only marginal evidence of a pre-trend violation for echocardiograms ($p=0.06$). In the post period, the estimates continue to show gradual long-run strengthening of planned care---higher adherence and greater specialist/diagnostic follow-up---while ER admissions remain near zero throughout. This restriction therefore confirms that the main findings are not driven by patients already managed by NPCP before the shock and are instead consistent with NPCP operating primarily as a post-discharge, longer-run care-coordination intervention (Figure~\ref{fig:sa121}).


\section{Conclusion}
\label{sec:Conclusion}

Using linked 2017--2023 administrative data, we study how care changes after a first CHF hospitalization and whether enrollment in a nurse-led chronic-care program (NPCP) alters that path. The index admission is a sharp clinical shock that triggers intensive discharge planning and a patient ``wake-up'' response. Consistent with this, preventive inputs---$\beta$-blocker adherence, cardiology follow-up, and echocardiographic monitoring---rise immediately even among patients who never enroll in NPCP, indicating a common post-hospitalization response rather than an incremental program effect. Relative to this common shock profile, our main finding is that NPCP appears to improve the persistence of preventive care. Under modern staggered-adoption event-study estimates, effects are small and imprecise in the first few post-discharge quarters but become positive and economically meaningful about six or more quarters after the shock. In that long-run window, enrollees maintain higher guideline-recommended management: adherence remains elevated, and cardiology follow-up and monitoring increase and persist. Thus, evaluations focused only on short horizons would understate the program's contribution to durable chronic-disease management.

By contrast, we find no clear evidence that NPCP reduces emergency room (ER) use. Estimated effects on ER visits remain near zero throughout the post-hospitalization period, consistent with offsetting mechanisms: improved outpatient management may prevent some acute exacerbations, while closer monitoring and patient education may also increase precautionary ER use.

A central interpretive issue is identification. Modern staggered-adoption DiD estimators yield a coherent pattern of rising preventive inputs, but our conservative HonestDiD sensitivity analysis shows that confidence intervals can include zero under modest deviations from parallel trends. We therefore interpret the long-run gains in preventive care as suggestive and consistent with the program's intended mechanism, while stopping short of claiming definitive causal effects under the strictest parallel-trends assumptions.

This pattern fits naturally within a health-capital framework. A CHF hospitalization raises the perceived marginal return to health investment and creates a short-lived ``discharge surge'' in provider attention and patient effort, helping explain the immediate improvement in adherence and follow-up for all patients. The relevant informational and coordination frictions are concrete barriers to sustained self-management and outpatient care: limited understanding of regimens and warning signs, uncertainty about whom to contact and when, difficulty scheduling and navigating visits and tests, fragmented communication across providers, and attention or forgetfulness constraints once symptoms stabilize. These frictions are lowest around the index hospitalization---when instructions are salient, appointments may be pre-booked, and monitoring is intense---but tend to re-emerge after discharge as supervision fades and patients return to routine life. Without support, initially higher investments may therefore decay over time (missed refills, delayed visits, fewer planned tests), even if preferences and clinical need remain unchanged. NPCP is designed to reduce these post-discharge frictions through education, proactive outreach, and care coordination, thereby lowering the effective cost and raising the productivity of sustained health investment. In this framework, the program's contribution is not the universal discharge surge, but the persistence of preventive behaviors and the accumulation of preventive management after the acute phase.

Our findings also align with evidence from other chronic-care settings. Nurse-coordinated interventions often strengthen chronic-disease management. Nurse-led diabetes programs significantly reduce hemoglobin A$_{1c}$ (Wang et al., 2019), intensive nurse management of hypertension improves long-run blood pressure control (Ito et al., 2024), and intensive nurse-led care after myocardial infarction improves medication adherence and cardiac self-care (Lizcano-\'Alvarez et al., 2023). In heart failure, specialized nurse-led clinics have increased adherence and substantially reduced readmissions (Wu et al., 2024). Institutional context also matters. Our study is set in Italy's universal-coverage NHS with strong primary-care gatekeeping, where GPs decide whether to enroll patients in chronic-disease management programs (De Belvis et al., 2024). We would therefore expect similar programs to perform similarly---or even more effectively---in comparably structured systems such as the UK's NHS. In more fragmented systems, or those with greater patient choice (e.g.\ the US), weaker GP coordination and different incentives may attenuate the realized impact of nurse-led coordination models.

We acknowledge several limitations. Low NPCP take-up reduces precision, especially at short horizons and in subgroup analyses by age, gender, and baseline comorbidity (MCDS). Our administrative data do not distinguish CHF-specific ER visits from other emergencies, and we do not observe program intensity or implementation fidelity across sites. More broadly, selection into enrollment cannot be fully ruled out, and our sensitivity checks show that even modest violations of parallel trends can overturn conventional statistical significance.

Despite these caveats, the overall pattern is informative. Nurse-led chronic care appears to strengthen long-run preventive health investment after a major CHF shock, consistent with a human-capital view of health, while its effect on acute care use remains uncertain. Future work should link NPCP exposure to harder clinical endpoints (readmissions, mortality) and costs, measure heterogeneity in implementation intensity, and exploit more plausibly exogenous variation (e.g.\ provider practice-style differences or staged rollout) to sharpen causal interpretation. Taken together, our evidence provides a credible basis for viewing nurse-led chronic care as a policy tool that can strengthen patients' long-run health-management inputs.

\appendix

\renewcommand{\thefigure}{A\arabic{figure}}
\renewcommand{\thetable}{A\arabic{table}}
\setcounter{figure}{0}
\setcounter{table}{0}

\section*{Appendix A: Structural Health-Capital Model and Event-Study Mapping}
\label{app:structural_model_short}

\setcounter{equation}{0}
\renewcommand{\theequation}{A\arabic{equation}}

This appendix provides a concise link between a Grossman-style health-capital model and
the event-study specification used in the paper. We show why an absorbing program like NPCP can generate dynamic effects on preventive investments (adherence, specialist visits, diagnostics) and potentially effects on ER use.

\paragraph{A1. Health transition with NPCP channels.}
Let $H_{i,t}$ denote the health stock at the start of period $t$. In the baseline,
\begin{equation}
H_{i,t+1} = (1-\delta_{i,t})\,H_{i,t} + \phi_t(g_i)\,I_{i,t},
\label{eq:A1_short}
\end{equation}
where $0<\delta_{i,t}<1$ is depreciation and $I_{i,t}$ is a composite index of health investment
(market inputs and patient effort). We model NPCP as an absorbing treatment $n_{i,t}\in\{0,1\}$
that can (i) reduce effective depreciation and (ii) raise input productivity:
\begin{equation}
H_{i,t+1} =
\bigl[1-\delta_{i,t}+\gamma_1 n_{i,t}\bigr]H_{i,t}
+\phi_t(g_i)\bigl[1+\gamma_2 n_{i,t}\bigr]I_{i,t},
\label{eq:A2_short}
\end{equation}
with $\gamma_1>0$ and $\gamma_2>0$.

\paragraph{A2. Linearized health deviations.}
Let $h_{i,t}\equiv H_{i,t}-\bar H_i$ be the deviation from the no-program steady state.
A first-order approximation of \eqref{eq:A2_short} around $n_{i,t}=0$ yields a stable linear law:
\begin{equation}
h_{i,t+1} = \lambda_H(g_i)\,h_{i,t} + \xi(g_i)\,n_{i,t} + u_{i,t},
\label{eq:A3_short}
\end{equation}
where $|\lambda_H(g_i)|<1$ summarizes health persistence, $\xi(g_i)$ is the composite shift in
the health transition induced by NPCP, and $u_{i,t}$ collects shocks/approximation error.
Intuitively, $\xi(g_i)$ bundles the ``depreciation'' channel ($\gamma_1$) and the ``productivity''
channel ($\gamma_2$), evaluated at baseline investment levels.

\paragraph{A3. Investment outcomes and ER use.}
The dynamic optimization problem (utility over consumption and health subject to the health law)
implies diminishing returns to investment. Rather than carry the full solution, we use a reduced-form
linearization consistent with the first-order condition:
\begin{equation}
I_{i,t}-\bar I_i = \alpha_H\,h_{i,t} + \alpha_N\,n_{i,t} + v_{i,t},
\label{eq:A4_short}
\end{equation}
where $\alpha_H$ captures how investment responds to health and $\alpha_N$ captures direct program
effects on investment costs/productivity (e.g., lower informational/behavioral costs).

Observable outcomes load on health and may also have a direct NPCP component. For an ``investment''
outcome $Y^{(m)}$ (adherence, visits, diagnostics), write:
\begin{equation}
Y^{(m)}_{i,t} = \alpha^{(m)}_i + \lambda^{(m)}_t + q_m\,h_{i,t} + r_m\,n_{i,t} + X'_{i,t}\theta^{(m)} + \varepsilon^{(m)}_{i,t}.
\label{eq:A5_short}
\end{equation}
For ER use, allow an ambiguous direct effect (e.g., monitoring/triage could increase precautionary ER
visits even if health improves):
\begin{equation}
Y^{(ER)}_{i,t} = \alpha^{(ER)}_i + \lambda^{(ER)}_t + \pi_1\,h_{i,t} - \pi_2\,n_{i,t} + X'_{i,t}\theta^{(ER)} + \varepsilon^{(ER)}_{i,t}.
\label{eq:A6_short}
\end{equation}

\paragraph{A4. Mapping to event-study coefficients under an absorbing treatment.}
Let $T_i$ be NPCP start time and define event time $k=t-T_i$. With absorbing treatment,
$n_{i,t}=1[k\ge 0]$. Under \eqref{eq:A3_short}, normalizing $h_{i,T_i-1}=0$ and ignoring shocks
for intuition, the implied health path is
\begin{equation}
h_{i,T_i+k} =
\xi(g_i)\sum_{s=0}^{k-1}\lambda_H(g_i)^s
=\xi(g_i)\,\frac{1-\lambda_H(g_i)^k}{1-\lambda_H(g_i)},
\qquad k\ge 0,
\label{eq:A7_short}
\end{equation}
and $h_{i,T_i+k}=0$ for $k<0$ under no anticipation.

Therefore, the event-study coefficient for an investment outcome satisfies, for $k\ge 0$,
\begin{equation}
\beta^{(m)}_k = r_m + q_m\,h_{i,T_i+k},
\label{eq:A8_short}
\end{equation}
and for ER use,
\begin{equation}
\beta^{(ER)}_k = (-\pi_2) + \pi_1\,h_{i,T_i+k}.
\label{eq:A9_short}
\end{equation}
This mapping implies dynamic effects that generally converge to a long-run impact (a new steady state) rather than mechanically widening over time. Lead coefficients ($k<0$) are predicted to be zero in the absence of anticipation/pre-trends.

\clearpage



\renewcommand{\thefigure}{B\arabic{figure}}
\renewcommand{\thetable}{B\arabic{table}}
\setcounter{figure}{0}
\setcounter{table}{0}

\section*{Appendix B: Data and Identification}
\label{app:Data_ID}

\begin{figure}[htbp]
    \centering
    \includegraphics[width=4.9in]{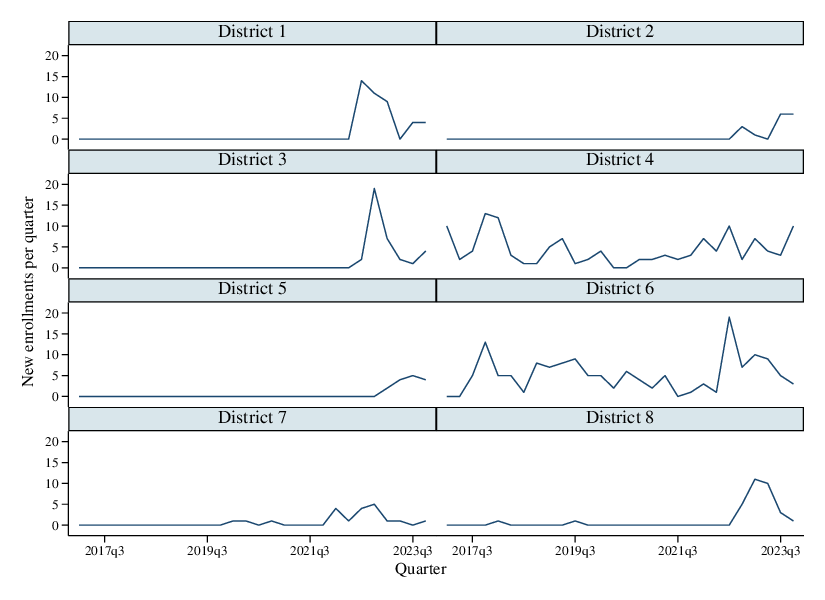}
    \caption{New enrollments into the nurse-led program by district and quarter, 2017–2023. Each panel plots intake\_dt, the number of patients who newly enroll in the current quarter (enroll\_q = quarter\_period). The variable is constructed from the CHF patient–quarter panel and aggregated to the district–quarter level.}
    \label{Fig:scatter_data}
\end{figure}

\begin{table}[htbp]\centering
\caption{GP enrollment propensity (leave-one-out) by NPCP enrollment status}
\label{tab:gp_leniency_loo}
\footnotesize
\begin{tabular}{l lrrrrrr}
\toprule
& & \multicolumn{2}{c}{Controls ($D=0$)} & \multicolumn{2}{c}{Treated ($D=1$)} & \multicolumn{2}{c}{Treated $-$ Control} \\
\cmidrule(lr){3-4}\cmidrule(lr){5-6}\cmidrule(lr){7-8}
Panel & Sample definition & $N$ & Mean (SD) & $N$ & Mean (SD) & Diff. & $t$-stat \\
\midrule
A & Full eligible sample
  & 15{,}097 & 0.0275 (0.0697)
  & 487 & 0.1490 (0.1399)
  & 0.1215 & 36.21 \\
B & Restrict to GPs with $n_{\text{gp}}\ge 10$
  & 12{,}896 & 0.0279 (0.0652)
  & 420 & 0.1421 (0.1110)
  & 0.1141 & 34.31 \\
\bottomrule
\end{tabular}

\vspace{0.2cm}
\begin{minipage}{0.98\linewidth}
\footnotesize
\textit{Notes:} $D$ indicates NPCP enrollment (treated $=1$). The table reports $L^{\text{LOO}}_{\text{gp}}$, a leave-one-out measure of GP $g$'s NPCP enrollment propensity for patient $i$:
$L^{\text{LOO}}_{\text{gp}} = \left(\sum_{j\in g} D_j - D_i\right)/(n_g - 1)$, i.e., the share of \emph{other} eligible patients assigned to the same GP who enroll (excluding patient $i$). The leave-one-out construction avoids mechanically embedding the patient's own enrollment status in the GP-level rate. Panel A uses the full eligible sample with observed baseline GP assignment. Panel B restricts to GPs with $n_{\text{gp}}\ge 10$ eligible patients to reduce noise from small GP panels. \emph{Diff.} is the treated-minus-control mean difference in $L^{\text{LOO}}_{\text{gp}}$ (in levels). $t$-statistics are from two-sample equality-of-means tests assuming equal variances; corresponding $p$-values are $<0.001$ in both panels. The large differences indicate strong sorting of enrolled patients toward higher-enrollment-propensity GPs, consistent with heterogeneity in GP practice style.
\end{minipage}
\end{table}

\clearpage



\renewcommand{\thefigure}{C\arabic{figure}}
\renewcommand{\thetable}{C\arabic{table}}
\setcounter{figure}{0}
\setcounter{table}{0}

\section*{Appendix C: Moving Average Graphs}
\label{App:DES_Graphs}

\begin{figure}[htbp]
	\centering
	\includegraphics[width=4.9in]{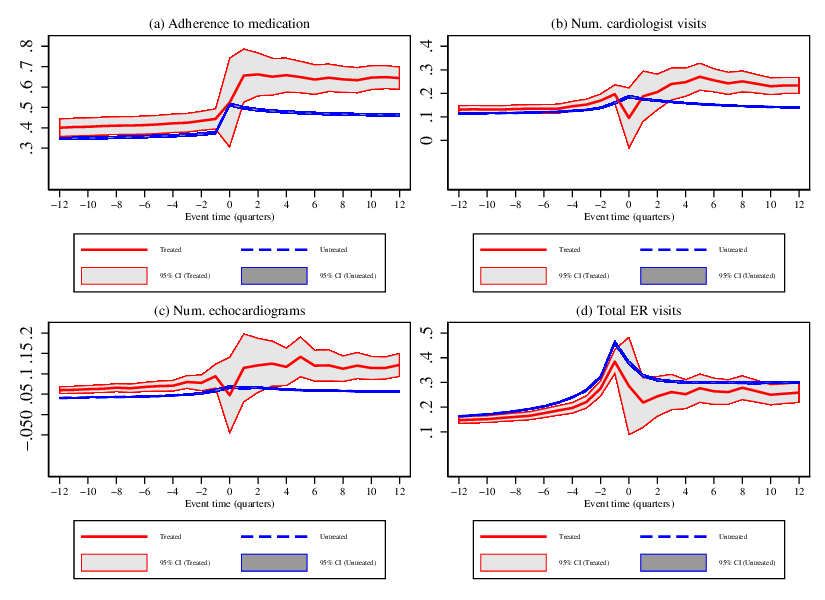}
	\caption{Moving-average event-time trajectories around the first CHF hospitalization ($t=0$) for panels (a)--(d): $\beta$-blocker adherence, cardiologist visits, echocardiograms, and ER visits. Lines plot forward and backward moving averages anchored at $t=0$ and $t=-1$; shaded bands report 95\% confidence intervals for treated and control groups over $\pm 12$ quarters. The sample includes treated and untreated patients and excludes the 121 patients who enrolled in NPCP before the index hospitalization.}
	\label{Fig:movAveexcl121}
\end{figure}


\begin{figure}[htbp]
	\centering
	\includegraphics[width=4.9in]{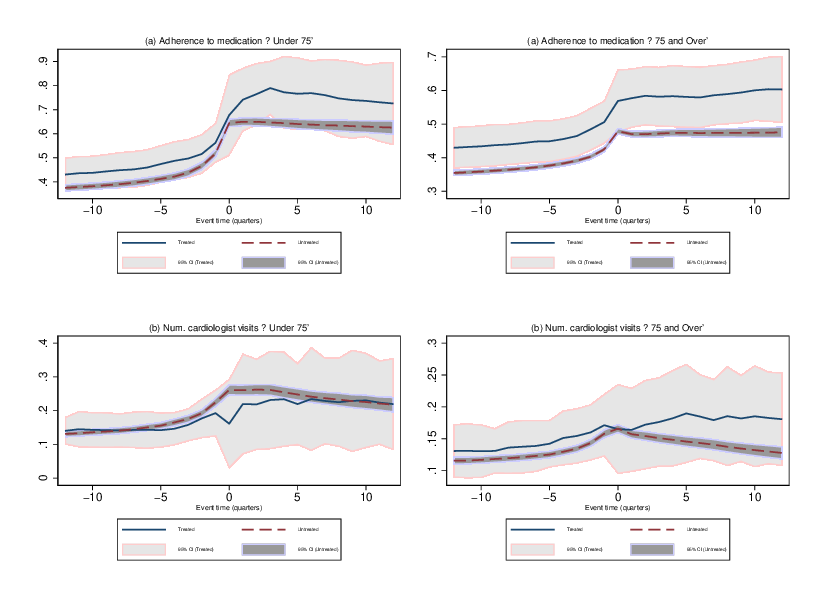}
        \caption{This figure shows moving-average outcome trajectories for patients under age 75 versus patients 75 and older. Panels (a) and (b) plot the paths for (a) $\beta$-blocker medication adherence and (b) cardiologist visits, centered at the quarter of the first CHF hospitalization (event time 0). Moving averages are computed forward from event time 0 and backward from event time -1. Shaded areas represent 95\% confidence intervals.}
    \label{Fig:scatter10}
\end{figure}


\begin{figure}[htbp]
	\centering
	\includegraphics[width=4.9in]{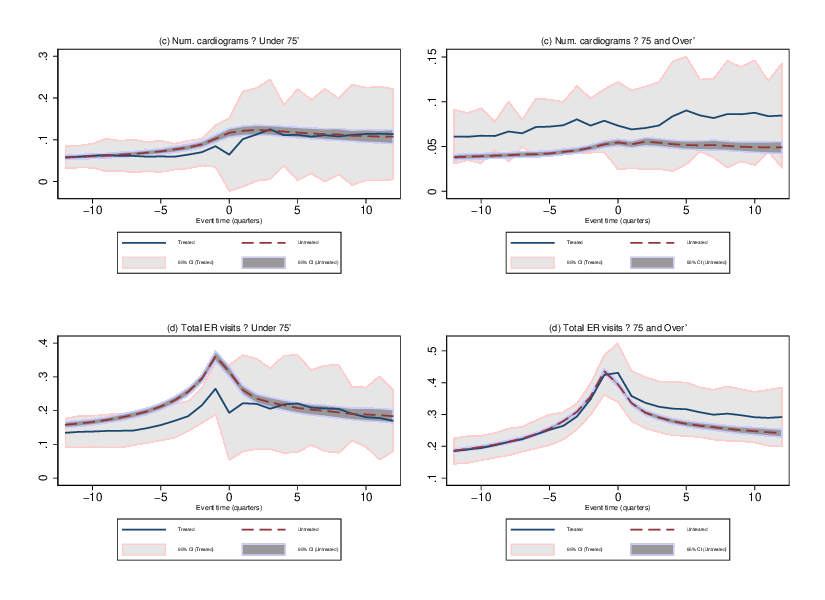}
        \caption{This figure shows moving-average trajectories for additional outcomes by age group (patients under 75 versus those 75 and over). Panels (c) and (d) plot the paths for (c) echocardiogram use and (d) total ER visits, with each series benchmarked at event time 0 (the quarter of the first hospitalization). Shaded bands denote 95\% confidence intervals around each series.}

	\label{Fig:scatter12}
\end{figure}
\clearpage

\begin{figure}[htbp]
	\centering
	\includegraphics[width=4.9in]{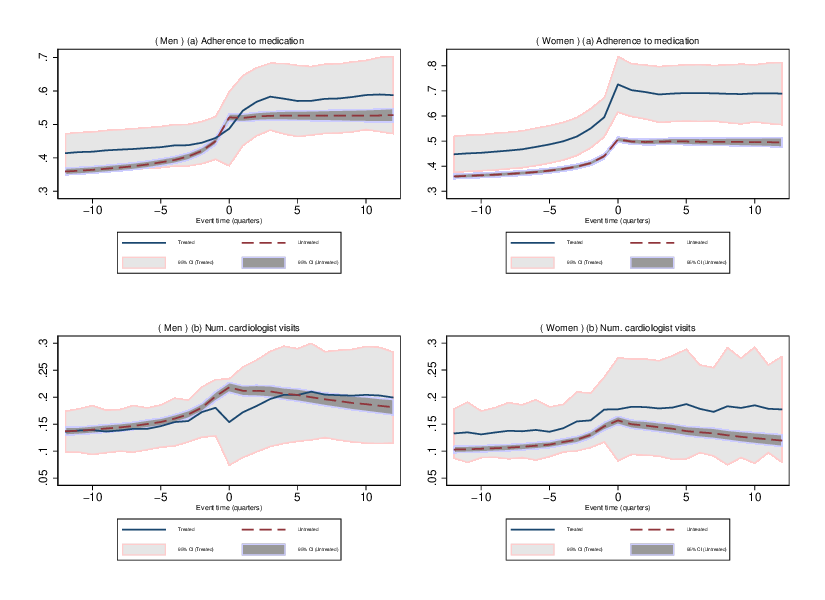}
        \caption{This figure shows moving-average outcome trajectories for gender of patients. Panels (a) and (b) plot the paths for (a) $\beta$-blocker medication adherence and (b) cardiologist visits, centered at the quarter of the first CHF hospitalization (event time 0). Moving averages are computed forward and backward from event time 0, and 95\% confidence intervals are indicated by the shaded areas.}
	\label{Fig:scatter20}
\end{figure}

\begin{figure}[htbp]
	\centering
	\includegraphics[width=4.9in]{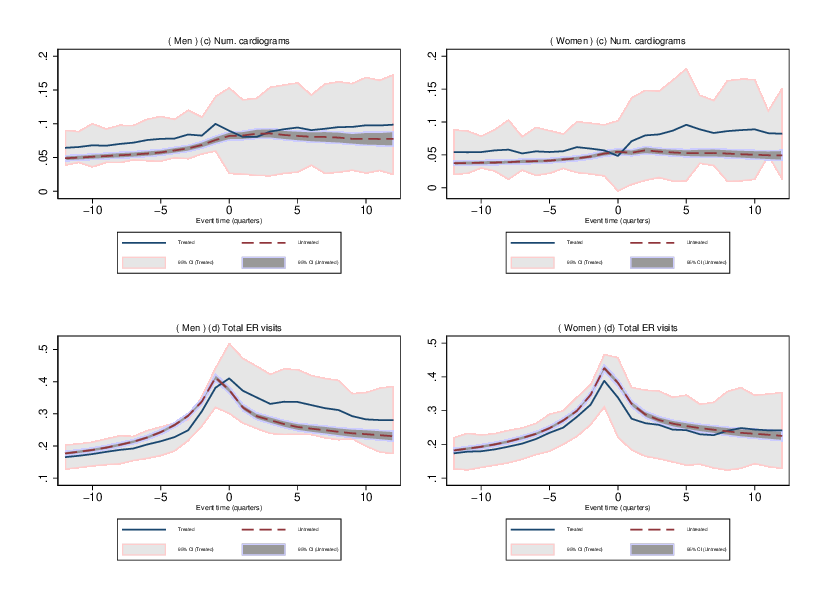}
        \caption{This figure shows moving-average outcome trajectories for gender of patients. Panels (c) and (d) plot the paths for (c) echocardiogram use and (d) total ER visits, with each series centered at event time 0 (the quarter of the first CHF hospitalization). Shaded bands denote 95\% confidence intervals for each trajectory.}
	\label{Fig:scatter21}
\end{figure}

\begin{figure}[htbp]
	\centering
	\includegraphics[width=4.9in]{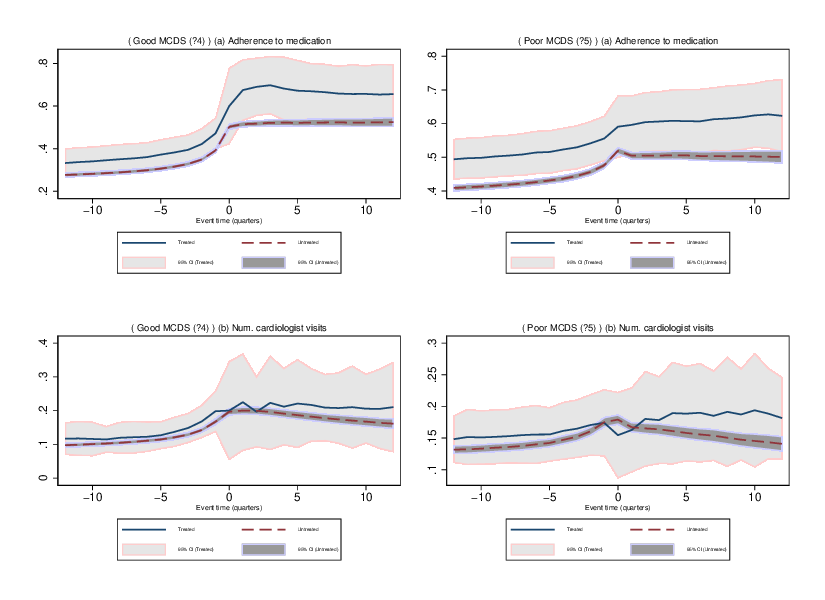}
        \caption{This figure shows moving-average outcome trajectories for patients with complex comorbidity (poor MCDS health status, with MCDS $\ge 5$ vs good MCDS health status, eith MCDS $\le 4$). Panels (a) and (b) plot the paths for (a) $\beta$-blocker medication adherence and (b) cardiologist visits, aligned at the quarter of the first CHF hospitalization (event time 0). Moving averages are computed forward and backward from event time 0, and shaded areas indicate 95\% confidence intervals.}
	\label{Fig:scatter22}
\end{figure}

\begin{figure}[htbp]
	\centering
	\includegraphics[width=4.9in]{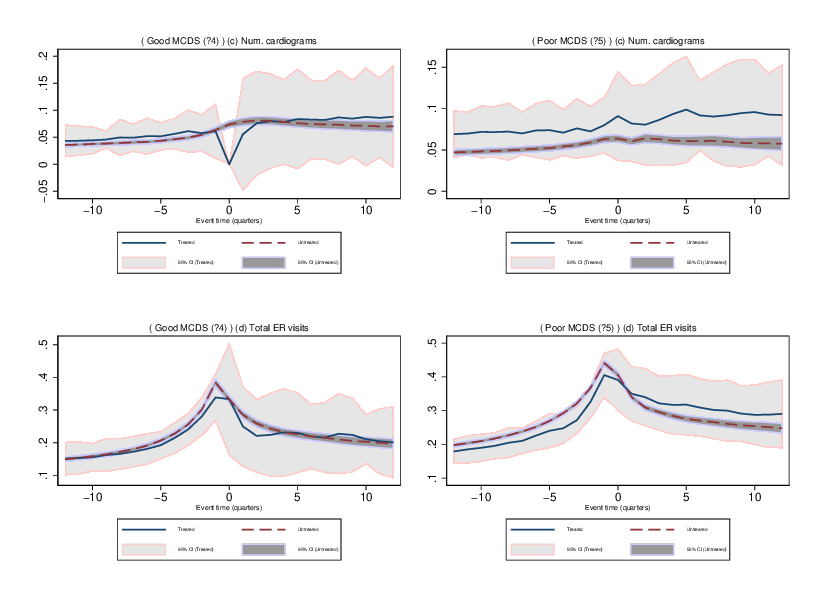}
        \caption{This figure shows moving-average trajectories for additional outcomes among patients with complex comorbidity (Poor MCDS health status, with MCDS $\ge 5$ vs good MCDS health status, eith MCDS $\le 4$). Panels (c) and (d) plot the paths for (c) echocardiogram use and (d) total ER visits, with each series centered at event time 0. Shaded bands denote 95\% confidence intervals for the outcomes.}
	\label{Fig:scatter23}
\end{figure}

\clearpage



\renewcommand{\thefigure}{D\arabic{figure}}
\renewcommand{\thetable}{D\arabic{table}}
\setcounter{figure}{0}
\setcounter{table}{0}

\section*{Appendix D: Additional Empirical Design Details}
\label{app:ignorable-censoring}

Our event-study analyses are indexed by \emph{event time relative to the first CHF hospitalization}, which we treat as a plausibly exogenous clinical shock. Because NPCP enrollment may occur either before or after hospitalization, some periods fall between the two dates. In those intervals, treatment status relative to hospitalization is ambiguous. To avoid this ambiguity, we censor outcomes in those periods.

\paragraph{Censoring rule.}
Let $T_i^H$ be the time of patient $i$'s first CHF hospitalization (the event-study anchor) and $T_i^E$ the time of NPCP enrollment. Define
\[
R_{it} =
\begin{cases}
1, & \text{if } t < \min\{T_i^H,T_i^E\} \;\; \text{or} \;\; t \ge \max\{T_i^H,T_i^E\},\\
0, & \text{otherwise}.
\end{cases}
\]
Equivalently:  
- If enrollment follows hospitalization, drop $t \in [T_i^H,\,T_i^E-1]$.  
- If enrollment precedes hospitalization, drop $t \in [T_i^E,\,T_i^H-1]$.  

Thus $R_{it}$ is a deterministic function of $(t,T_i^H,T_i^E)$. We restrict analysis to $\{(i,t):R_{it}=1\}$, i.e.\ periods where treatment status relative to hospitalization is uniquely defined.

\paragraph{Key idea.}
This censoring is \emph{missing by design}: it depends only on observed enrollment and hospitalization dates, not on potential outcomes. Hence it does not bias identification (Little and Rubin, 2019).

\begin{lemma}
If $R_{it}$ is a deterministic function of $(T_i^H,T_i^E)$ and conditional parallel trends hold on the kept support, then the Sun and Abraham (2021) interaction-weighted event-study estimand remains identified on $\{R_{it}=1\}$.
\end{lemma}

\begin{proof}
(1) By construction, $R_{it}$ depends only on $(T_i^H,T_i^E,t)$, which are fully observed.  
(2) Therefore, conditional on these dates, $R_{it}$ is independent of potential outcomes.  
(3) The Sun and Abraham estimator compares treated units to not-yet-treated units at the same relative time $\ell$. Dropping ambiguous periods removes intervals with undefined treatment status but does not alter these identifying contrasts.  
(4) Hence, on the restricted support $\{R_{it}=1\}$, the interaction-weighted estimator continues to recover the average treatment effect on the treated (ATT) relative to hospitalization. 
\end{proof}

\paragraph{Takeaway.}
We drop only those periods where treatment status is undefined. Since this censoring depends solely on observed dates and not on outcomes, it is ignorable for identification. The event-study estimates therefore remain valid.



\renewcommand{\thefigure}{E\arabic{figure}}
\renewcommand{\thetable}{E\arabic{table}}
\setcounter{figure}{0}
\setcounter{table}{0}

\section*{Appendix E: Estimation Graphs}
\label{App:EST_Graphs}


\begin{figure}[htbp]
  \centering
  \includegraphics[width=.82\linewidth,trim=0pt 5pt 0pt 0pt, clip]{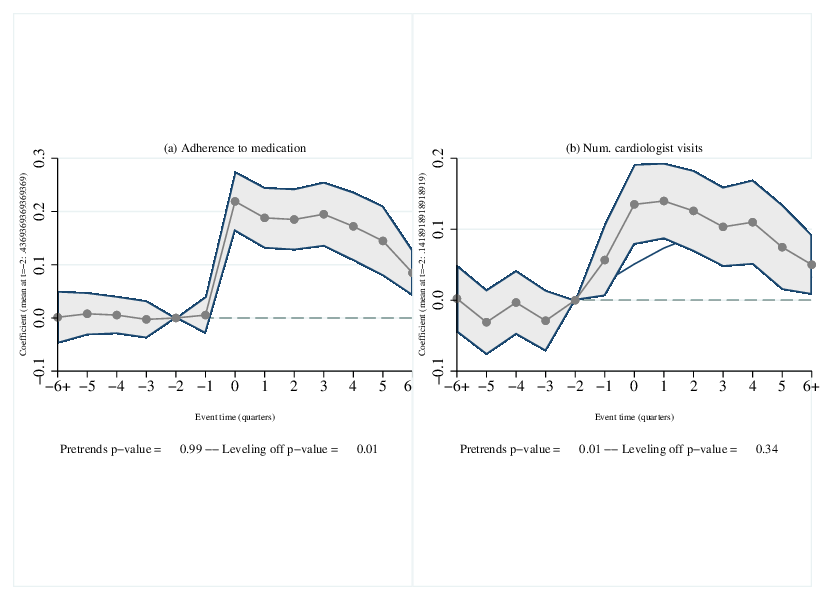}

  \vspace{0.4em}

  \includegraphics[width=.82\linewidth,trim=0pt 0pt 0pt 0pt, clip]{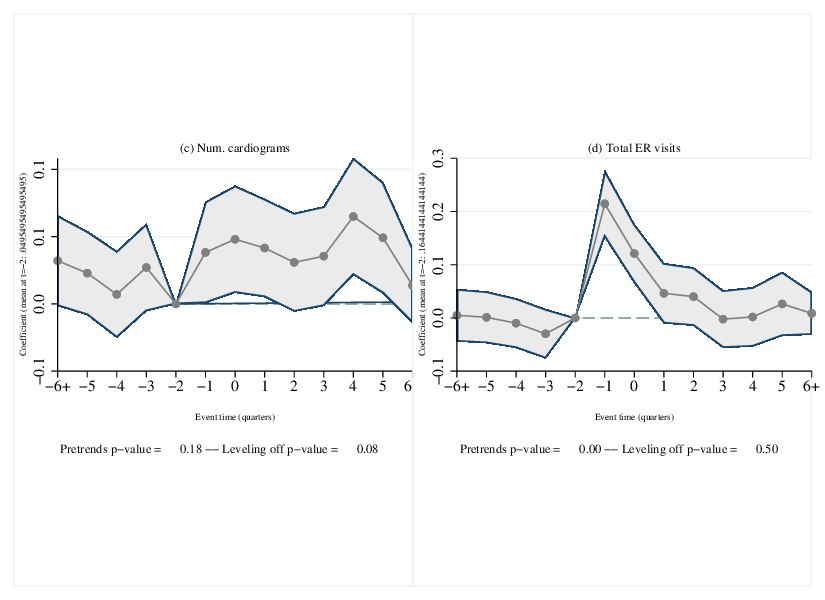}

\caption{Adverse-shock effects of the first CHF hospitalization for treated patients, estimated using the Sun--Abraham (2021) event-study design. Event time is measured in quarters relative to hospitalization (\(\tau=0\); reference period \(\tau=-2\)). Panels (a)--(d) show \(\beta\)-blocker adherence, cardiologist visits, echocardiograms, and ER visits for patients who enroll in NPCP. Points denote estimates, bands show 95\% confidence intervals, and the long-run bin pools \(\tau \ge 6\).}
  \label{fig:sa2T}
\end{figure}

\begin{figure}[htbp]
  \centering
  \includegraphics[width=.82\linewidth,trim=0pt 5pt 0pt 0pt, clip]{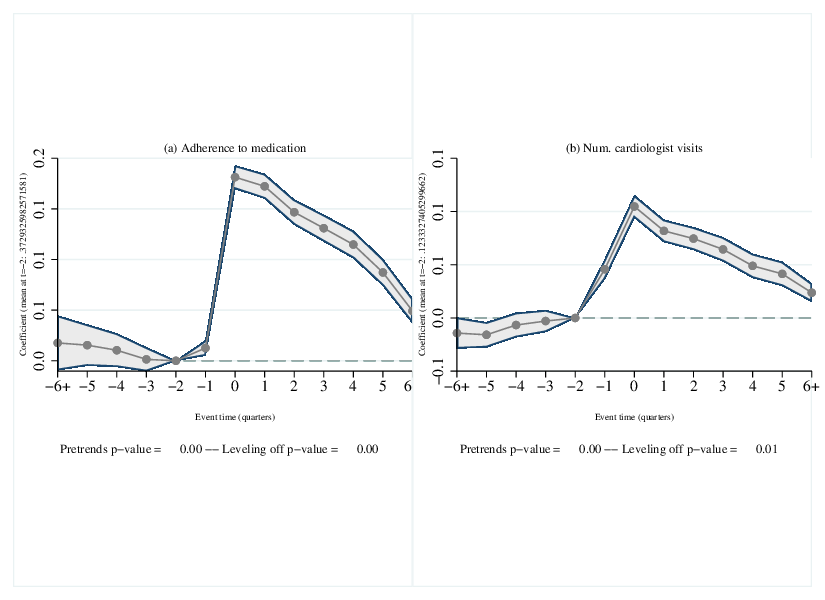}

  \vspace{0.4em}

  \includegraphics[width=.82\linewidth,trim=0pt 0pt 0pt 0pt, clip]{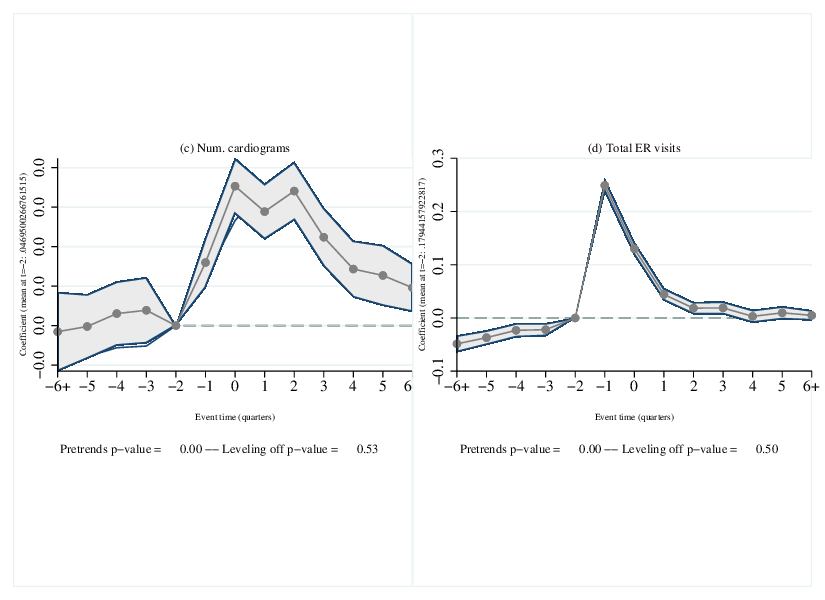}

\caption{Adverse-shock effects of the first CHF hospitalization for untreated patients, estimated using the Sun--Abraham (2021) event-study design. Event time is measured in quarters relative to hospitalization (\(\tau=0\); reference period \(\tau=-2\)). Panels (a)--(d) show \(\beta\)-blocker adherence, cardiologist visits, echocardiograms, and ER visits for patients who never enroll in NPCP. Points denote estimates, bands show 95\% confidence intervals, and the long-run bin pools \(\tau \ge 6\).}
  \label{fig:sa2UT}
\end{figure}









\begin{figure}[htbp]
  \centering
  \includegraphics[width=.82\linewidth,trim=0pt 5pt 0pt 0pt, clip]{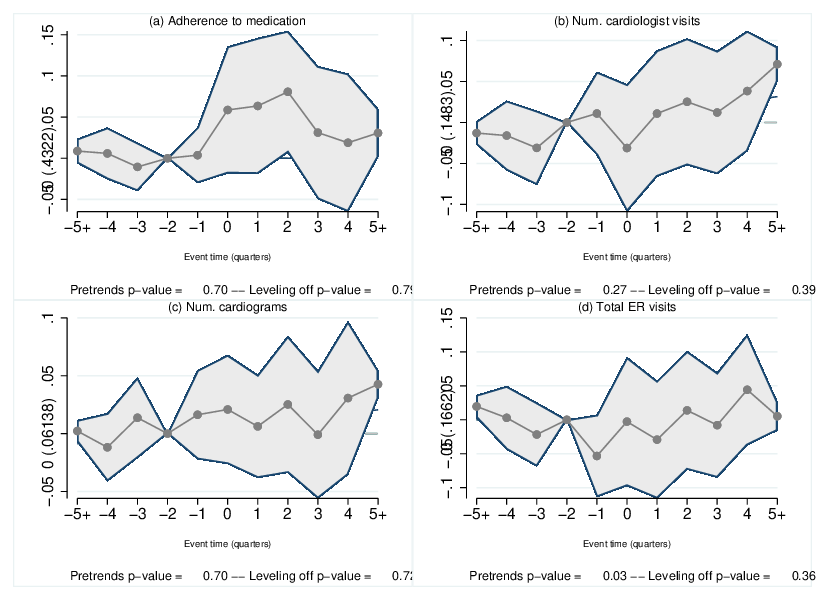}

\caption{NPCP effects via Sun--Abraham (2021) event study. Event time in quarters around first CHF hospitalization (\(\tau=0\); ref = \(-2\)). Panels (a--d): \(\beta\)-blocker adherence, cardiologist visits, echocardiograms, ER visits. Points = estimates; bands = 95\% CIs; long-run bin pools \(\tau \ge 5\).}
  \label{fig:sa4}
\end{figure}



\begin{figure}[htbp]
  \centering
  \includegraphics[width=.82\linewidth,trim=0pt 5pt 0pt 0pt, clip]{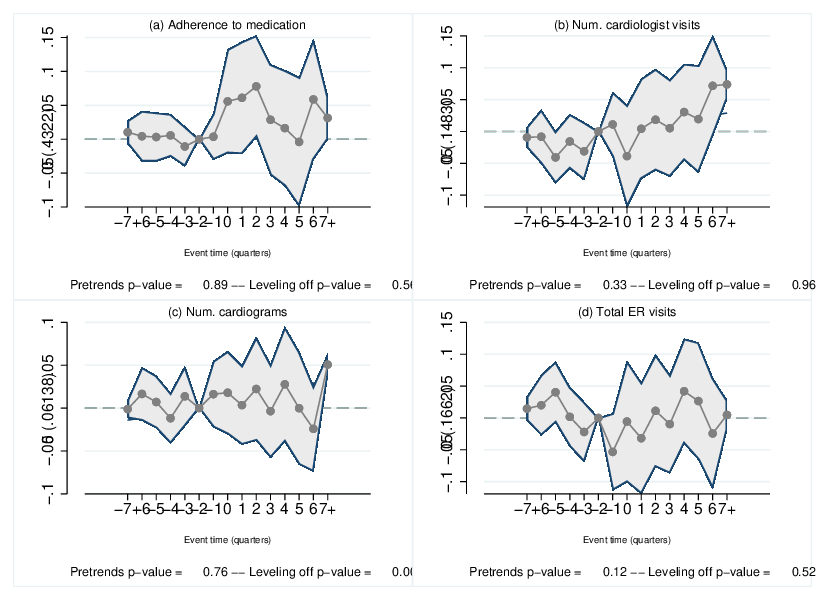}

\caption{NPCP effects via Sun--Abraham (2021) event study. Event time in quarters around first CHF hospitalization (\(\tau=0\); ref = \(-2\)). Panels (a--d): \(\beta\)-blocker adherence, cardiologist visits, echocardiograms, ER visits. Points = estimates; bands = 95\% CIs; long-run bin pools \(\tau \ge 7\).}
  \label{fig:sa7}
\end{figure}


\begin{figure}[htbp]
  \centering
  \includegraphics[width=.82\linewidth,trim=0pt 5pt 0pt 0pt, clip]{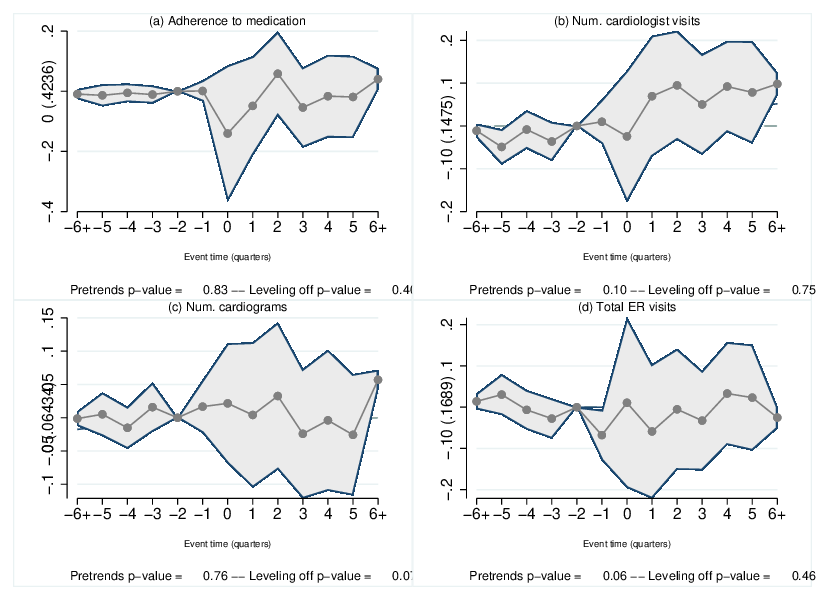}

\caption{NPCP effects via Sun--Abraham (2021) event study excludes 121 patients who enrolled before the first hospitalization experienced. Event time in quarters around first CHF hospitalization (\(\tau=0\); ref = \(-2\)). Panels (a--d): \(\beta\)-blocker adherence, cardiologist visits, echocardiograms, ER visits. Points = estimates; bands = 95\% CIs; long-run bin pools \(\tau \ge 6\).}
  \label{fig:sa121}
\end{figure}

\begin{figure}[htbp]
    \centering
    \includegraphics[width=0.85\textwidth]{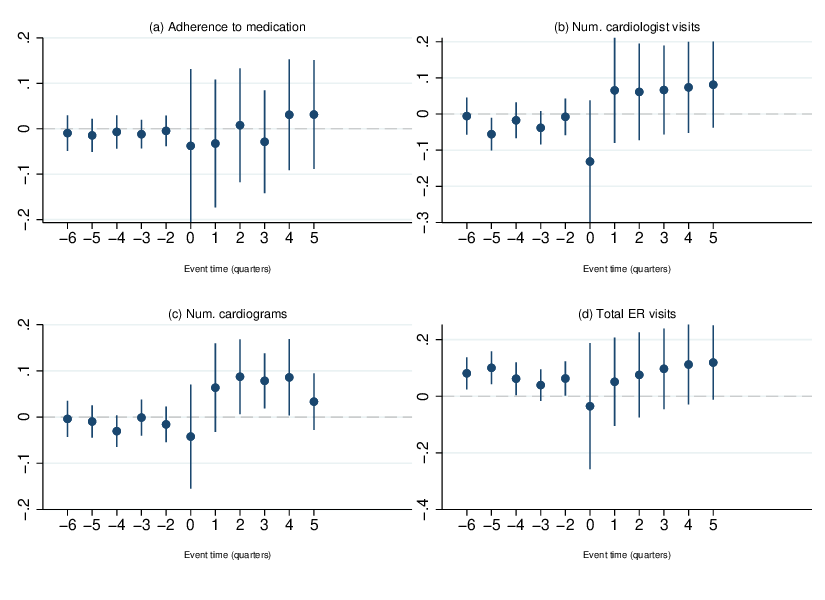}
\caption{This figure shows dynamic treatment effect estimates using the doubly robust estimator of Callaway and Sant’Anna (2021) with inverse probability weighting. Panels (a)–(d) correspond to effects on $\beta$-blocker adherence, cardiologist visits, echocardiograms, and ER visits, respectively.}
    \label{fig:cs_main}
\end{figure}
\clearpage



\begin{figure}[htbp]
  \centering
  \includegraphics[width=0.75\linewidth,trim=0pt 0pt 0pt 0pt,clip]{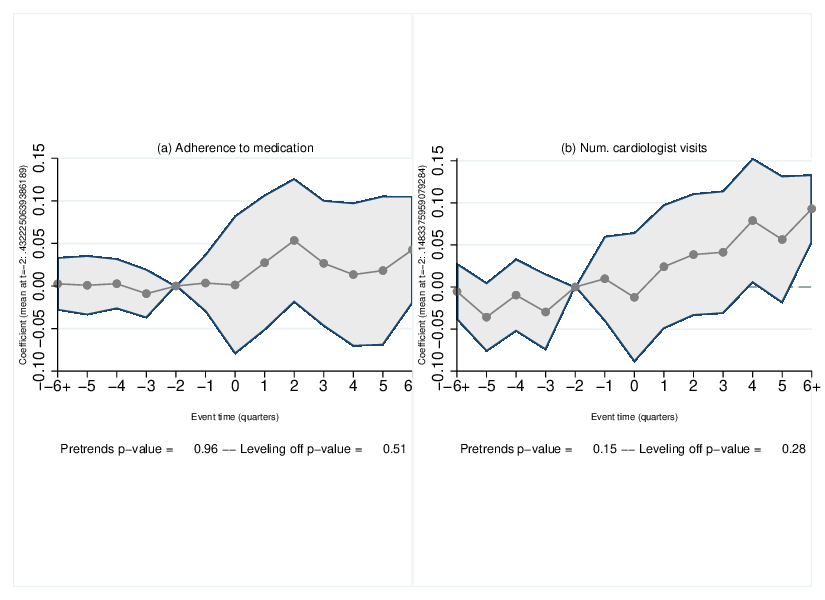}

  \vspace{0.4em}

  \includegraphics[width=0.82\linewidth,trim=0pt 0pt 0pt 0pt,clip]{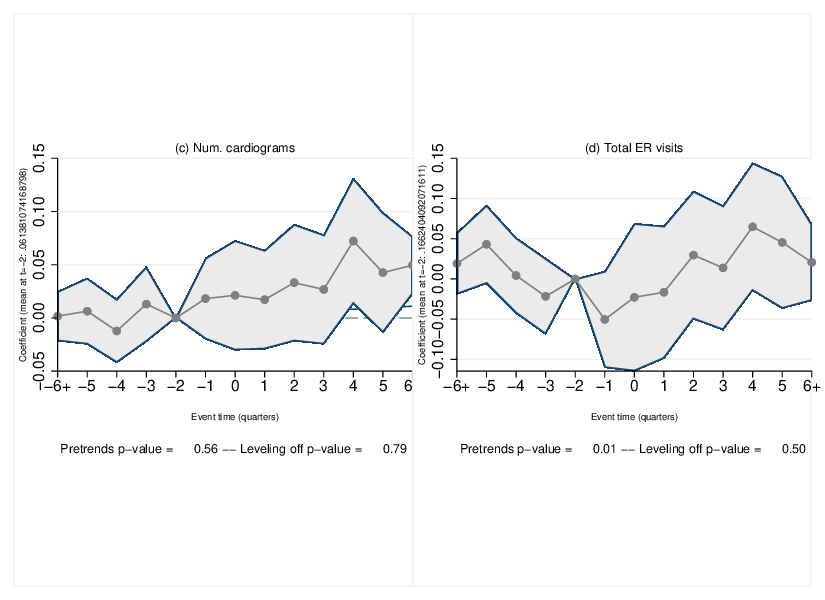}

  \caption{NPCP enrollment effects under \textbf{homogeneous treatment effects} via Sun--Abraham (2021) event study. Event time in quarters since first CHF hospitalization ($\tau=0$; ref $=-2$). Panels (a–d): $\beta$-blocker adherence, cardiologist visits, echocardiograms, ER visits. Circles $=$ estimates; shading $=$ 95\% pointwise CIs; dashed $=$ 95\% uniform sup-$t$ bands.}
  \label{fig:ES_homo}
\end{figure}


\begin{figure}[htbp]
  \centering
  \includegraphics[width=.82\linewidth,trim=6pt 6pt 6pt 6pt,clip]{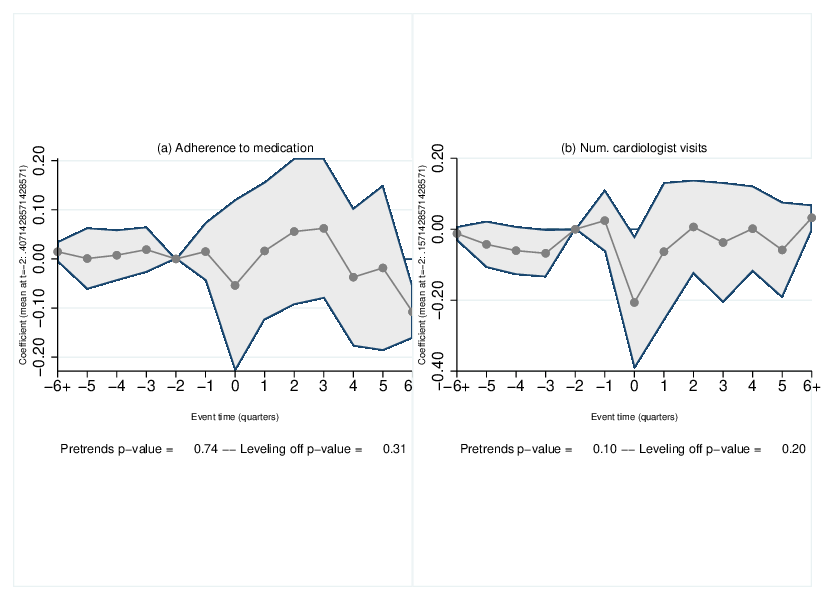}
  \includegraphics[width=.82\linewidth,trim=6pt 6pt 6pt 6pt,clip]{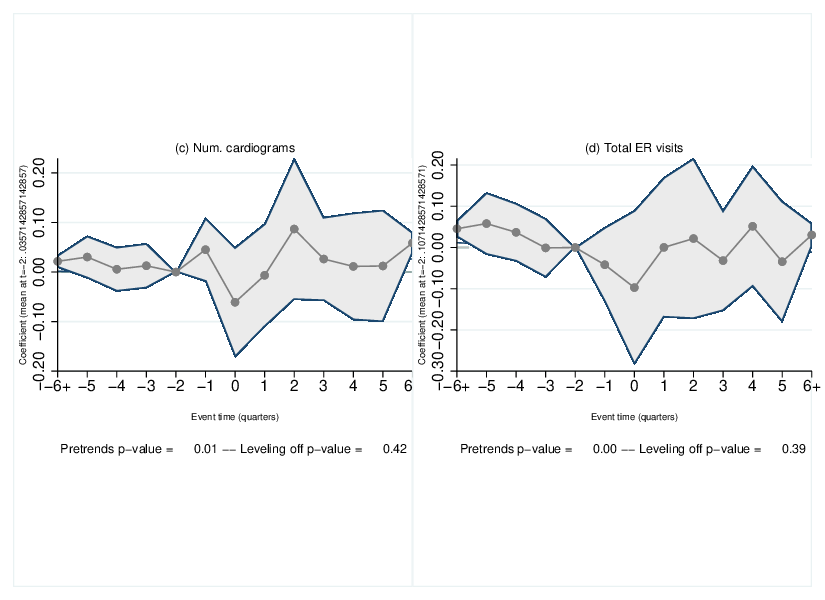}

  \caption{NPCP effects for \textbf{ patients age} $<\!75$ via Sun--Abraham (2021) event study. Event time in quarters since first CHF hospitalization ($\tau=0$; ref $=-2$). Panels (a--d): $\beta$-blocker adherence (PDC $\ge 75\%$), any cardiologist visit, any echocardiogram, any ER visit. Points $=$ estimates; bars $=$ 95\% CIs (SEs clustered by patient). Reported “Pre-trends” and “Leveling-off” are Wald-test $p$-values; y-axis parentheses show the reference-period mean ($\tau=-2$).}
  \label{fig:ES_SA_lt75}
\end{figure}

\begin{figure}[htbp]
    \centering
    \includegraphics[width=\textwidth]{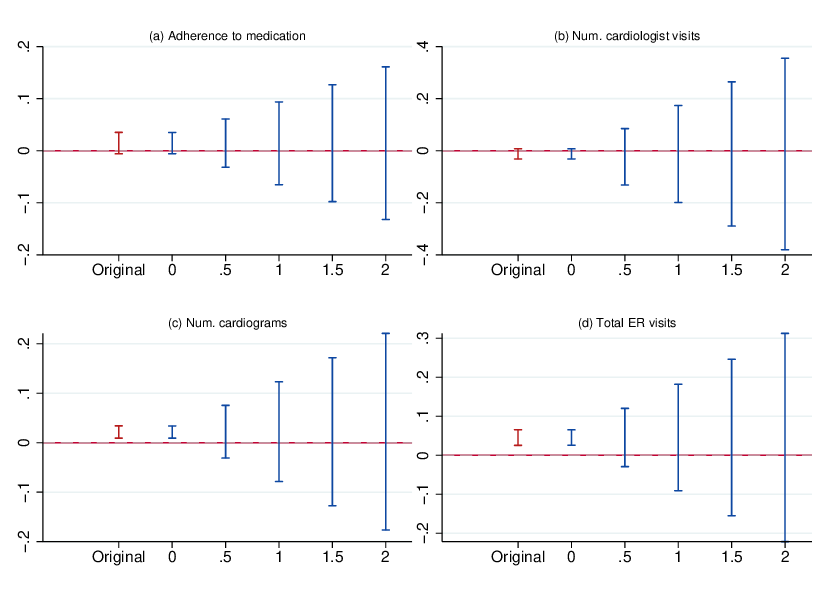}
\caption{HonestDiD sensitivity for \textbf{patients age} $<75$. Panels (a--d): $\beta$-blocker adherence, cardiology visit, echocardiogram, ER visit. 95\% CIs shown across sensitivity values $M$ (allowed deviation from parallel pre-trends); $M=0$ reproduces the Sun--Abraham baseline.}
    \label{fig:honestdid_lt75}
\end{figure}


\begin{figure}[htbp]
  \centering
  \includegraphics[width=.82\linewidth,trim=0pt 0pt 0pt 0pt,clip]{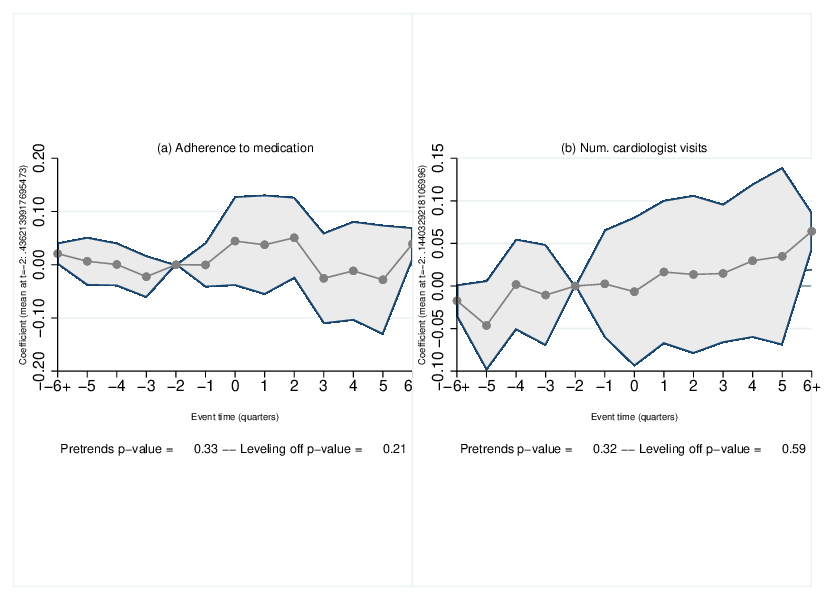}
  \includegraphics[width=.82\linewidth,trim=0pt 0pt 0pt 0pt,clip]{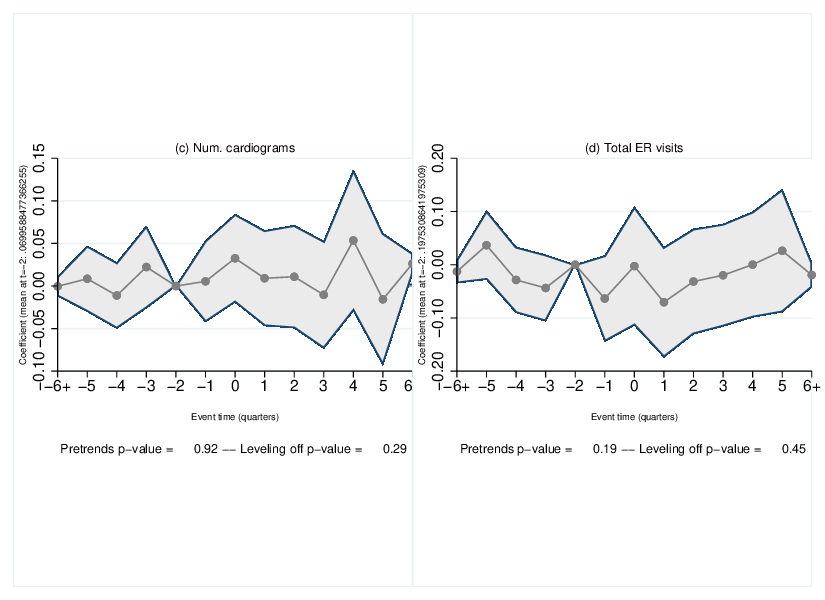}
\caption{NPCP effects for \textbf{patients age} $\ge 75$ via Sun--Abraham (2021) event study. Event time since first CHF hospitalization ($\tau=0$; ref $=-2$). Panels (a--d): $\beta$-blocker adherence (PDC $\ge 75\%$), cardiology visit, echocardiogram, ER visit. Points with 95\% CIs (SEs clustered by patient). “Pre-trends”/“Leveling-off”: Wald $p$-values; y-axis parentheses show the reference-period mean.}
  \label{fig:ES_SA_ot75}
\end{figure}

\begin{figure}[htbp]
    \centering
    \includegraphics[width=\textwidth]{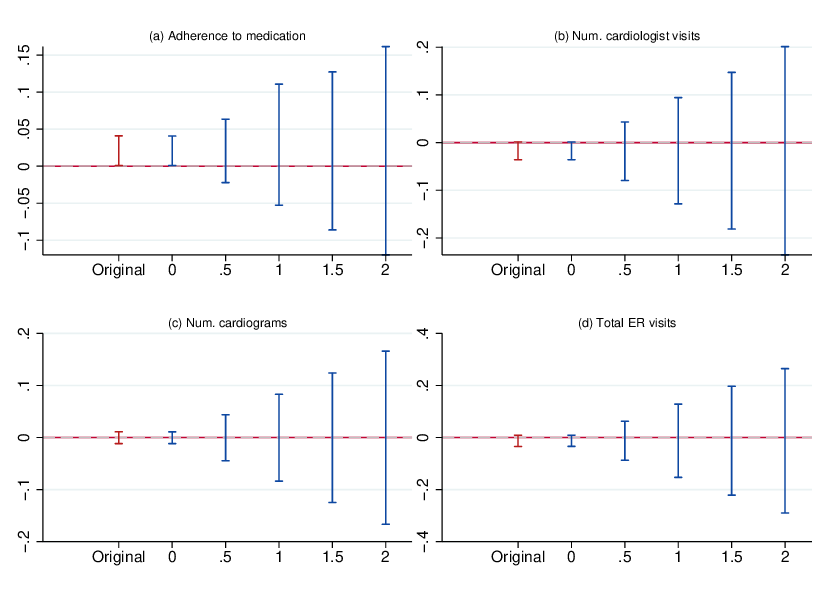}
\caption{HonestDiD sensitivity for \textbf{patients age} $\ge75$. Panels (a--d): $\beta$-blocker adherence, cardiology visit, echocardiogram, ER visit. 95\% CIs across sensitivity $M$ (deviation from parallel pre-trends); $M=0$ equals the Sun--Abraham baseline.}
    \label{fig:honestdid_ge75}
\end{figure}

\begin{figure}[htbp]
  \centering
  \includegraphics[width=.82\linewidth,trim=0pt 0pt 0pt 0pt,clip]{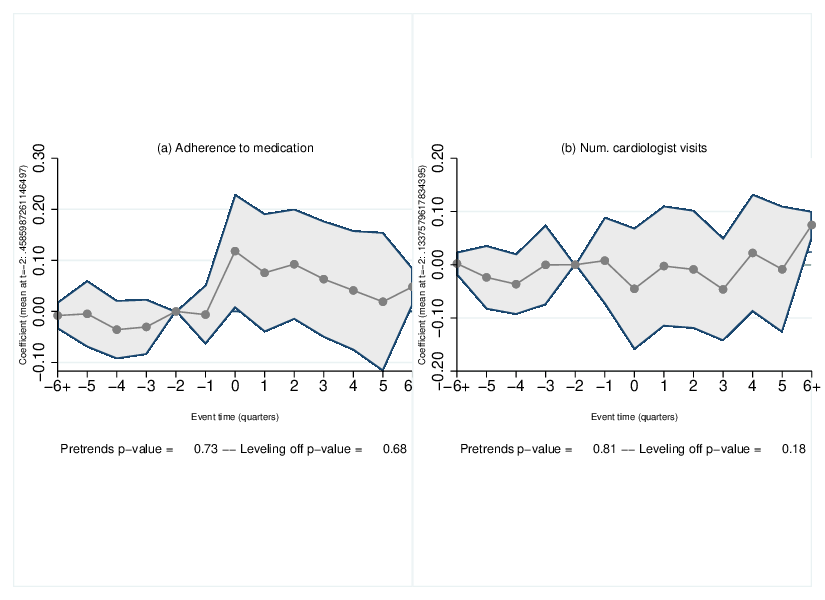}
  \includegraphics[width=.82\linewidth,trim=0pt 0pt 0pt 0pt,clip]{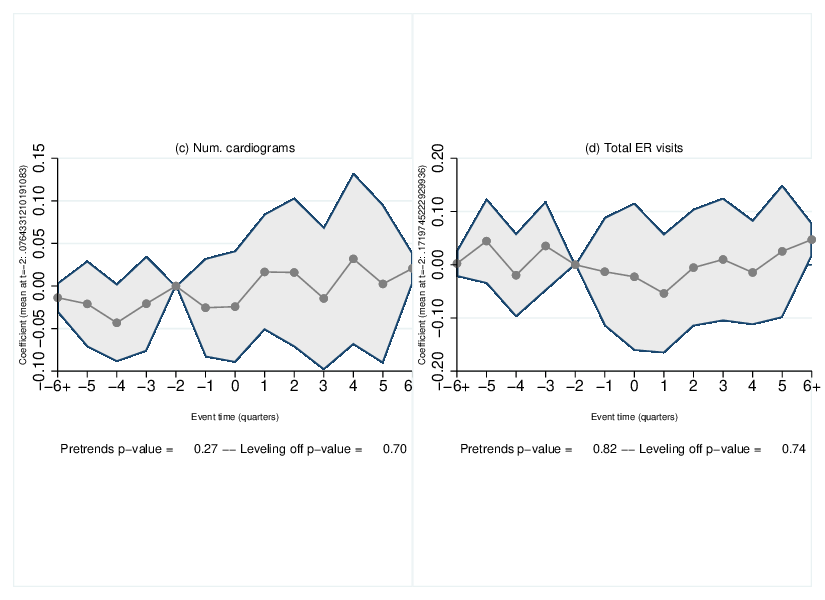}
\caption{NPCP effects for \textbf{female patients} via Sun--Abraham (2021) event study. Event time since first CHF hospitalization ($\tau=0$; ref $=-2$). Panels (a--d): $\beta$-blocker adherence (PDC $\ge 75\%$), cardiology visit, echocardiogram, ER visit. Points with 95\% CIs (SEs clustered by patient). “Pre-trends”/“Leveling-off”: Wald $p$-values; parentheses show the reference-period mean.}
  \label{fig:ES_SA_fem}
\end{figure}

\begin{figure}[htbp]
    \centering
    \includegraphics[width=\textwidth]{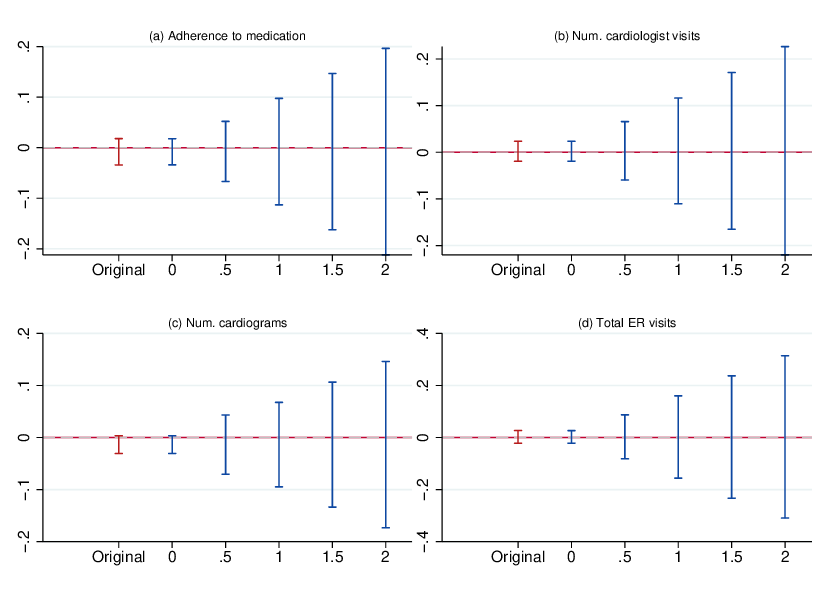}
\caption{HonestDiD sensitivity for \textbf{female patients}. Panels (a--d): $\beta$-blocker adherence, cardiology visit, echocardiogram, ER visit. 95\% CIs across sensitivity $M$ (allowed deviation from parallel pre-trends); $M=0$ equals the Sun--Abraham baseline.}
    \label{fig:honestdid_female}
\end{figure}


\begin{figure}[htbp]
  \centering
  \includegraphics[width=.82\linewidth,trim=0pt 0pt 0pt 0pt,clip]{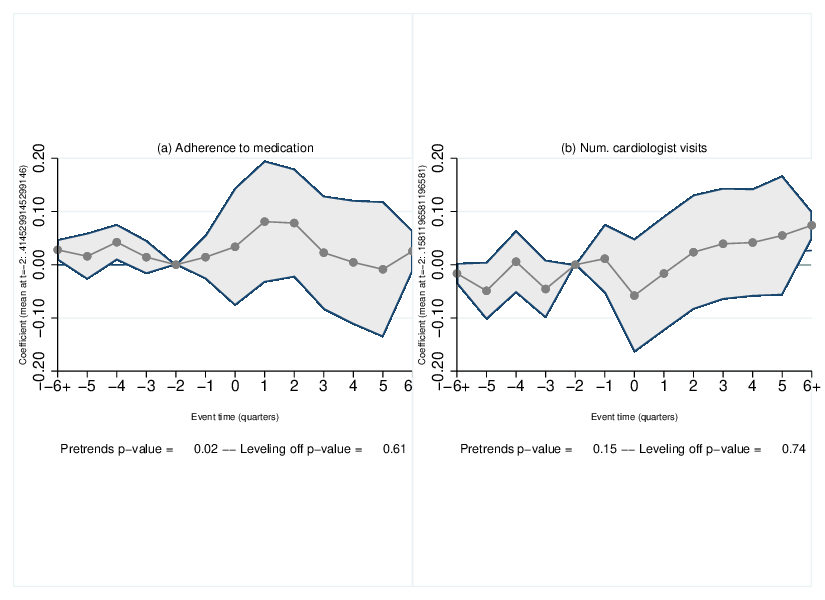}
  \includegraphics[width=.82\linewidth,trim=0pt 0pt 0pt 0pt,clip]{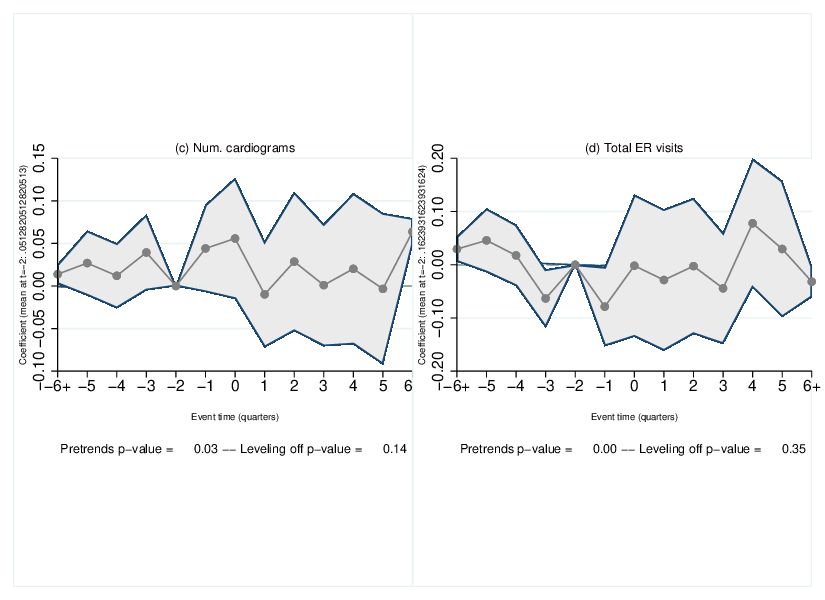}
\caption{NPCP effects for \textbf{male patients} via Sun--Abraham (2021) event study. Event time since first CHF hospitalization ($\tau=0$; ref $=-2$). Panels (a--d): $\beta$-blocker adherence (PDC $\ge 75\%$), cardiology visit, echocardiogram, ER visit. Points with 95\% CIs (SEs clustered by patient). “Pre-trends”/“Leveling-off”: Wald $p$-values; parentheses show the reference-period mean.}
  \label{fig:ES_SA_m}
\end{figure}

\begin{figure}[htbp]
    \centering
    \includegraphics[width=\textwidth]{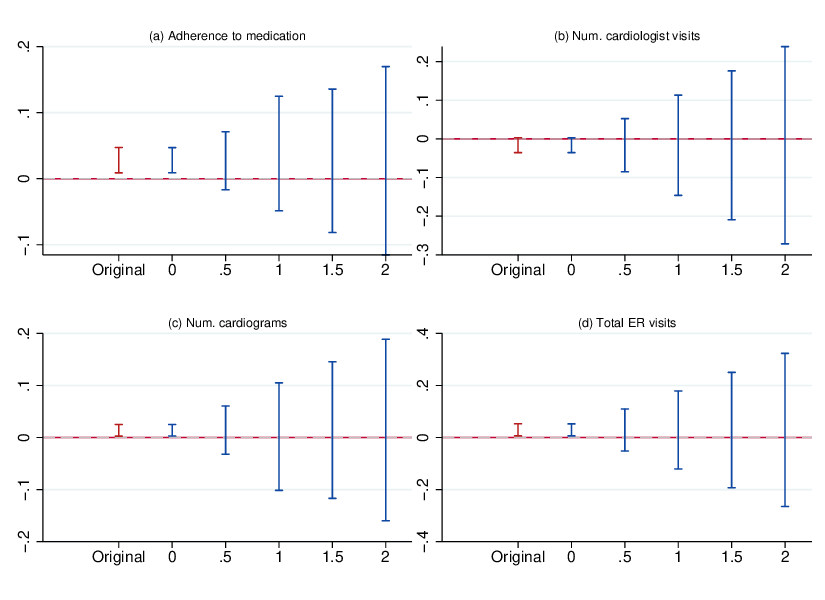}
\caption{HonestDiD sensitivity for \textbf{male patients}. Panels (a--d): $\beta$-blocker adherence, cardiology visit, echocardiogram, ER visit. 95\% CIs across sensitivity $M$ (deviation from parallel pre-trends); $M=0$ equals the Sun--Abraham baseline.}
    \label{fig:honestdid_male}
\end{figure}


\begin{figure}[htbp]
  \centering
  \includegraphics[width=.82\linewidth,trim=0pt 0pt 0pt 0pt,clip]{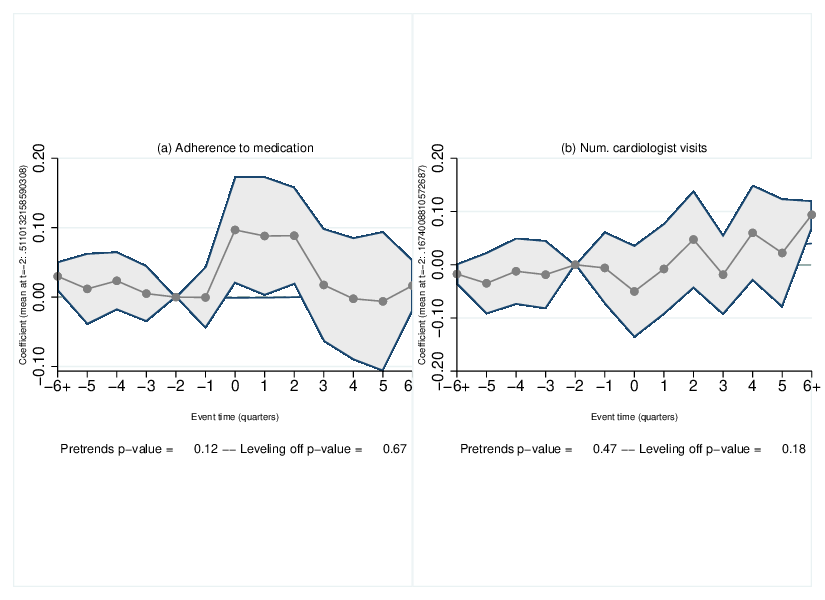}
  \includegraphics[width=.82\linewidth,trim=0pt 0pt 0pt 0pt,clip]{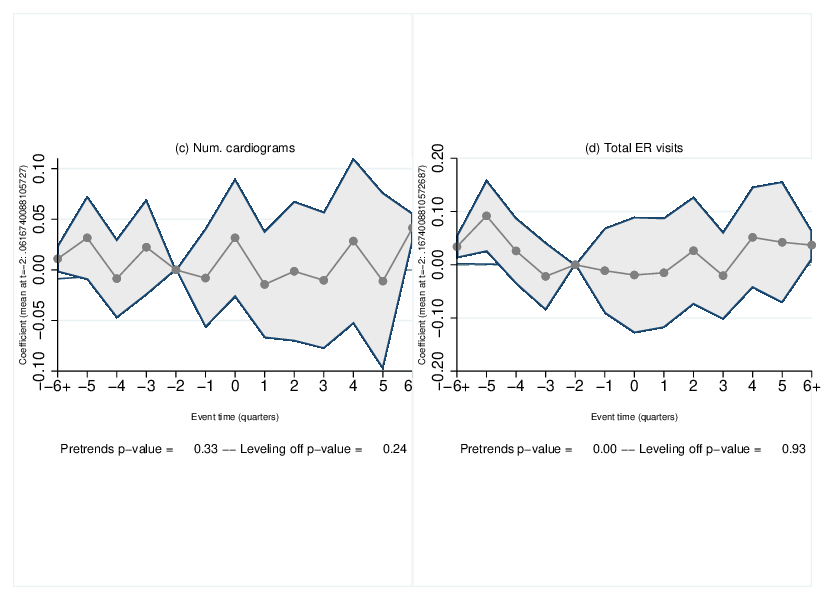}
\caption{NPCP effects for \textbf{high clinical complexity (MCDS $\ge5$)} via Sun--Abraham (2021) event study. Event time since first CHF hospitalization ($\tau=0$; ref $=-2$). Panels (a--d): $\beta$-blocker adherence (PDC $\ge75\%$), cardiology visit, echocardiogram, ER visit. Points $=$ IW ATTs with 95\% CIs (SEs clustered by patient). “Pre-trends”/“Leveling-off”: Wald $p$-values; parentheses show the reference-period mean.}
  \label{fig:ES_SA_MCDS_poor}
\end{figure}

\begin{figure}[htbp]
    \centering
    \includegraphics[width=\textwidth]{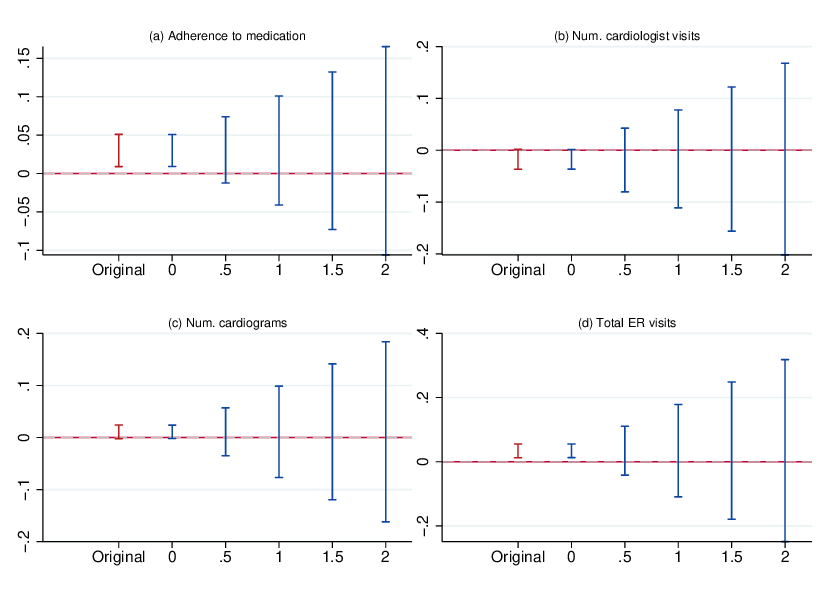}
\caption{HonestDiD sensitivity for \textbf{high comorbidity (MCDS $\ge5$)}. Panels (a--d): $\beta$-blocker adherence, cardiology visit, echocardiogram, ER visit. 95\% CIs across sensitivity $M$ (departure from parallel pre-trends); $M=0$ equals the Sun--Abraham baseline.}
    \label{fig:honestdid_MCDS_poor}
\end{figure}


\begin{figure}[htbp]
  \centering
  \includegraphics[width=.82\linewidth,trim=0pt 0pt 0pt 0pt,clip]{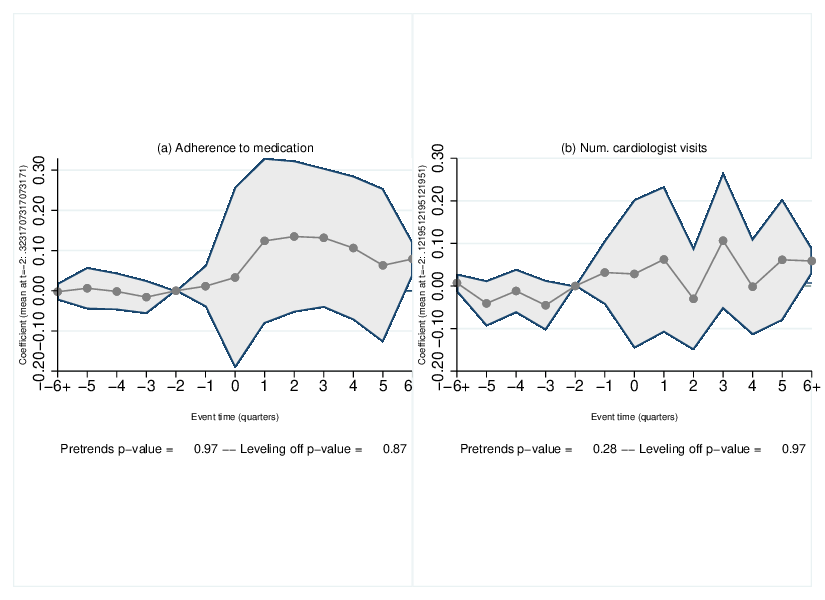}
  \includegraphics[width=.82\linewidth,trim=0pt 0pt 0pt 0pt,clip]{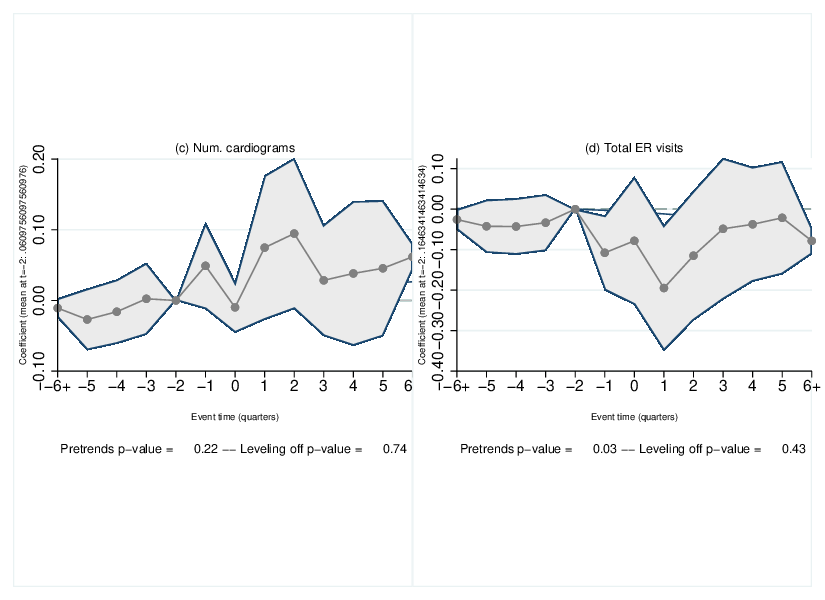}
\caption{NPCP effects for \textbf{lower clinical complexity (MCDS $\le4$)} via Sun--Abraham (2021) event study. Event time since first CHF hospitalization ($\tau=0$; ref $=-2$). Panels (a--d): $\beta$-blocker adherence (PDC $\ge75\%$), cardiology visit, echocardiogram, ER visit. Points with 95\% CIs (SEs clustered by patient). “Pre-trends”/“Leveling-off”: Wald $p$-values; parentheses show the reference-period mean.}
  \label{fig:ES_SA_MCDS_good}
\end{figure}

\begin{figure}[htbp]
    \centering
    \includegraphics[width=\textwidth]{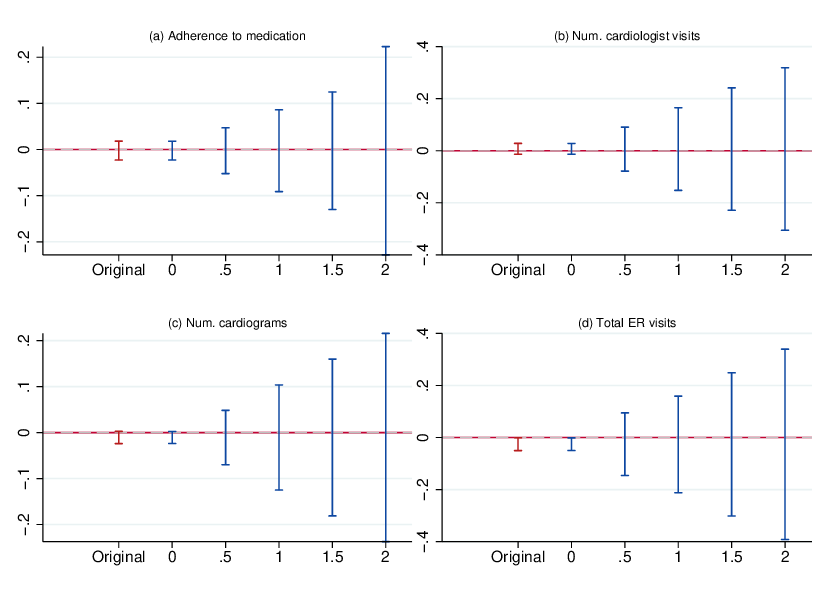}
\caption{HonestDiD sensitivity for \textbf{low comorbidity (MCDS $\le4$)}. Panels (a--d): $\beta$-blocker adherence, cardiology visit, echocardiogram, ER visit. 95\% CIs across sensitivity $M$ (deviation from parallel pre-trends); $M=0$ equals the Sun--Abraham baseline.}
    \label{fig:honestdid_MCDS_good}
\end{figure}

\clearpage



\renewcommand{\thefigure}{F\arabic{figure}}
\renewcommand{\thetable}{F\arabic{table}}
\setcounter{figure}{0}
\setcounter{table}{0}

\section*{Appendix F: Tables}
\label{Appendix F: Tables}

\begin{table}[htbp]\centering
\caption{Minimum Detectable Effects (MDE): Average Post Effect (Event Times 0--5)}
\label{tab:mde_avgpost}
\begin{tabular}{lccc}
\toprule
Outcome & Post window & SE (avg post) & MDE (avg post) \tabularnewline
\midrule
Medication adherence (beta-blockers) & p0--p5 &    0.0179 &    0.0503 \tabularnewline
Cardiologist visits & p0--p5 &    0.0160 &    0.0448 \tabularnewline
Echocardiograms & p0--p5 &    0.0126 &    0.0352 \tabularnewline
ER visits & p0--p5 &    0.0184 &    0.0515 \tabularnewline
\bottomrule
\end{tabular}
\vspace{0.2em}
\begin{minipage}{0.99\linewidth}\footnotesize
Notes: The table reports the standard error (SE) of the average post-treatment effect over event times 0--5 (inclusive) from the saved Sun--Abraham event-study specification. 
The MDE is computed as $\left(z_{1-\alpha/2} + z_{\text{power}}\right) \times \text{SE}$ with $\alpha = 0.05$ (two-sided) and power = 0.80. 
The implied multiplier equals  2.802.
\end{minipage}
\end{table}


\begin{table}[htbp]
\centering
\small
\setlength{\tabcolsep}{5pt}
\caption{This table reports joint tests for the presence of pre-treatment trends in each outcome. For each outcome (adherence, cardiologist visits, echocardiograms, ER visits), the table shows the F-statistic and p-value from a joint null hypothesis test that all pre-treatment event-time coefficients are zero (based on the Callaway–Sant’Anna, 2021 method). A higher p-value indicates a failure to reject the null of no pre-trend (i.e., parallel trends cannot be rejected). ** denotes significance at the 5\% level.}
\label{tab:pre-trend_CS}
\begin{tabular}{lrrrr}
\toprule
&  Adherence (PDC$>75\%$) & Cardiologist visit & Echocardiogram & ER visit \\
\midrule
$F$-statistic & 1.6511 & 8.8929 & 3.8136 & 30.2204$^{**}$ \\
$p$-value     & 0.8950 & 0.1134 & 0.5766 & 0.0000$^{**}$ \\
\bottomrule
\end{tabular}
\begin{flushleft}
\footnotesize\textit{} 
\end{flushleft}
\end{table}

\begin{table}[htbp]
\centering
\caption{This table reports Sun--Abraham (2021) interaction-weighted event-study coefficients estimating the effect of NPCP enrollment. Coefficients are indexed by event time relative to the quarter of the first CHF hospitalization ($t=0$) and are reported separately for pre- and post-hospitalization periods, including a pooled $6+$ quarters post-hospitalization bin. Outcomes are $\beta$-blocker adherence, cardiologist visits, echocardiograms, and ER visits.}
\begin{center}
\resizebox{\textwidth}{!}{%
\begin{tabular}{lcccc}
\toprule
& \textbf{Adherence to medication} & \textbf{Num. cardiologist Visit} & \textbf{Num. echocardiogram} & \textbf{Tot. access to ER}  \\
\midrule
6+ periods                    &       0.009   &      -0.010  &       0.002   &       0.016   \\
&     (0.008)   &     (0.006)   &     (0.005)   &     (0.009)   \\
5th period before treatment   &       0.003  &      -0.041*  &       0.007   &       0.042  \\
 &     (0.018)   &     (0.020)   &     (0.016)  &     (0.024)   \\
4th period before treatment   &       0.006   &      -0.016   &      -0.011   &       0.003   \\
&     (0.016)&      (0.020)&       (0.015)&      (0.024)  \\
3rd period before treatment   &      -0.011   &      -0.031  &       0.014&      -0.022  \\
&     (0.015)&      (0.022)&      (0.018)&      (0.024)   \\
1st period before treatment   &       0.004   &       0.011   &       0.016&        -0.053   \\
 &     (0.017)&       (0.025)&       (0.020)&      (0.031)   \\
0   period assign treatment   &       0.056&      -0.042&         0.020   &      -0.009  \\
 &     (0.039)&     (0.040)&     (0.025)&       (0.048)   \\
1st period after treatment    &       0.061&       0.002&        0.005   &      -0.035  \\
  &     (0.042)&       (0.039)&      (0.024)&      (0.045)   \\
2nd period after treatment    &       0.078*  &       0.016  &       0.024&         0.008  \\
 &     (0.038)&     (0.039)&      (0.031)&      (0.045)   \\
3rd period after treatment    &       0.029&        0.003&      -0.002&       -0.013   \\
 &     (0.041)&      (0.038)&     (0.028)&      (0.039)   \\
4th period after treatment    &       0.017&        0.028&       0.029   &       0.038   \\
  &     (0.043)&     (0.037)&      (0.034)&     (0.042)  \\
5th period after treatment    &      -0.004&        0.018&        0.002&         0.023   \\
  &     (0.048)&       (0.042)&       (0.034)&      (0.047)   \\
6+ periods                    &       0.032*  &       0.073***&       0.046***&       0.003   \\
  &     (0.015)&      (0.007)&     (0.006)&     (0.011)  \\

\bottomrule
\end{tabular}
}
\parbox{163mm} {\footnotesize {\em Standard errors in parentheses. $^{*}$ p$<$0.1, $^{**}$ p$<$0.05, $^{***}$ p$<$0.01.}
}
\label{tab:EST2}
\end{center}

\end{table}

\begin{table}[htbp]
\centering
\caption{Event-study estimates of the shock effect around the first CHF hospitalization (adverse shock) using the Sun--Abraham (2021) interaction-weighted estimator. Coefficients are reported by event time relative to the index quarter for $\beta$-blocker adherence, cardiologist visits, echocardiograms, and ER visits.}
\begin{center}
\resizebox{\textwidth}{!}{%
\begin{tabular}{lcccc}
\toprule
& \textbf{Adherence to medication} & \textbf{Num. cardiologist Visit} & \textbf{Num. echocardiogram} & \textbf{Tot. access to ER}  \\
\midrule
6+ periods     &      -0.010** &      -0.014***&      -0.002   &      -0.054***\\
            &     (0.004)   &     (0.003)   &     (0.002)   &     (0.004)   \\
5th period before treatment   &      -0.006   &      -0.014***&      -0.005*  &      -0.044***\\
      &     (0.003)   &     (0.003)   &     (0.002)  &     (0.004)   \\
4th period before treatment   &      -0.004   &      -0.005   &      -0.002   &      -0.032***\\
    &     (0.003)   &     (0.003)   &     (0.002)   &     (0.004)   \\
3rd period before treatment   &      -0.005   &      -0.007*  &      -0.002   &      -0.026***\\
     &     (0.002)   &     (0.003)   &     (0.002)  &     (0.004)   \\
1st period before treatment   &       0.015***&       0.041***&       0.013***&       0.265***\\
  &     (0.003)  &     (0.004)   &     (0.003)   &     (0.005)   \\
0   period assign treatment   &       0.152***&       0.072***&       0.021***&       0.186***\\
  &     (0.004)   &     (0.004)   &     (0.003)  &     (0.005)  \\
1st period after treatment    &       0.134***&       0.055***&       0.018***&       0.071***\\
 &     (0.005)   &     (0.004)   &     (0.003)   &     (0.005)   \\
2nd period after treatment    &       0.133***&       0.057***&       0.025***&       0.046***\\
 &     (0.005)   &     (0.004)  &     (0.003)  &     (0.005)   \\
3rd period after treatment    &       0.131***&       0.050***&       0.015***&       0.047***\\
&     (0.006)  &     (0.005)   &     (0.003)   &     (0.006)   \\
4th period after treatment    &       0.127***&       0.041***&       0.010** &       0.039***\\
 &     (0.006)   &     (0.005)   &     (0.003)   &     (0.006)   \\
5th period after treatment    &       0.115***&       0.032***&       0.005   &       0.041***\\
 &     (0.007)  &     (0.005)   &     (0.004)  &     (0.006)   \\
6+ periods                    &       0.104***&       0.012*  &      -0.001  &       0.042***\\
 &     (0.008)   &     (0.006)  &     (0.004)   &     (0.007)  \\
\bottomrule
\end{tabular}
}
\parbox{163mm} {\footnotesize {\em Standard errors in parentheses. $^{*}$ p$<$0.1, $^{**}$ p$<$0.05, $^{***}$ p$<$0.01.}
}
\label{tab:EST3}
\end{center}

\end{table}
\clearpage

\section*{Acknowledgements}
This research is part of a joint project of the Local Health Authority (LHA) of Romagna and of the Department of Economics of the University of Bologna “Evaluation of alternative organisational solutions for the management of chronic patients in the LHA of Romagna”. The authors thank Roberto Grilli, Head of the Unit “Evaluative Research and Health Services Policy” for his help in the analysis. The authors also thank Simona Rosa and Giulia Guidi for their support in the preparation of the data. The views expressed remain exclusively those of the authors.

\section*{Declaration}
We confirm that the work described has not been published previously, except in the form of a registered report, in accordance with the journal’s policy on multiple, redundant, or concurrent publication. The article is not under consideration for publication elsewhere. The article’s publication is approved by all authors and, tacitly or explicitly, by the responsible authorities where the work was carried out. If accepted, the article will not be published elsewhere in the same form, in English or in any other language, including electronically, without the written consent of the copyright-holder.

\section*{Data statement}
Data used in this study were obtained from the administrative datasets of the LHA of Romagna. The anonymized dataset covers the health consumption of CHF patients who were residents in the LHA of Romagna between 2017 and 2023 and provides information on patients’ characteristics, hospital services, ambulatory services, pharmaceutical consumption and general practitioners (GPs). This information is collected quarterly, and includes 17835 individuals, with 495936 observations across an average of 27.80 quarters. Data are not publicly available, but were made available to the authors, under license, from the LHA of Romagna after anonymization.

\section*{Generative AI declaration}
During the preparation of this manuscript, the authors used ChatGPT (OpenAI) to assist with language editing and text organization. The authors reviewed and edited all suggestions and take full responsibility for the final content.

\section*{Disclosure of interest}
The authors declare that they have no known financial or non-financial competing interests that could have influenced the work reported in this paper.

\section*{Funding}
This work was supported by the European Union – Next Generation EU – Age-It under Grant B83C22004800006; and by the Italian Complementary National Plan – DARE under Grant B53C2200645001.

\section*{References}
\begin{description}

\item Bahit, M.C., Korjian, S., Daaboul, Y., Baron, S., Bhatt, D.L., Kalayci, A., Chi, G., Nara, P., Shaunik, A. and Gibson, C.M., 2023. Patient adherence to secondary prevention therapies after an acute coronary syndrome: a scoping review.
\emph{Clinical therapeutics}, 45, 1119--1126. \DOI{10.1016/j.clinthera.2023.08.011}.

\item Callaway, B., Sant’Anna, P.H.C., 2021. Difference-in-differences with multiple time periods. J. Econometrics 225, 200–230. \DOI{10.1016/j.jeconom.2020.12.001}.

\item Chaudhry, S.I., Mattera, J.A., Curtis, J.P., Spertus, J.A., Herrin, J., Lin, Z., Phillips, C.O., Hodshon, B.V., Cooper, L.S., Krumholz, H.M., 2010. Telemonitoring in patients with heart failure. N. Engl. J. Med. 363, 2301–2309. \\
\DOI{10.1056/NEJMoa1010029}.

\item Chyn, E., Frandsen, B. and Leslie, E., 2025. Examiner and judge designs in economics: a practitioner’s guide.
\emph{Journal of Economic Literature}, 63, 401--439.

\item Conferenza Stato-Regioni (2024). Intesa sull’“Ipotesi di Accordo Collettivo Nazionale per la disciplina dei rapporti con i medici di medicina generale– Triennio 2019–2021” (Rep. Atti n. 51/CSR, 4 April 2024; agreement signed 8 February 2024).

\item Darden, M., 2017. Smoking, expectations, and health: a dynamic stochastic model of lifetime smoking behavior.
\emph{Journal of Political Economy}, 125, 1465--1522. \\
\DOI{10.1086/693394}.

\item De Belvis, A.G., Meregaglia, M., Morsella, A., Adduci, A., Perilli, A., Cascini, F., Solipaca, A., Fattore, G., Ricciardi, W., D’Agostino, M. and Maresso, A., 2024. Italy:
\emph{health system summary}, 2024.

\item de Chaisemartin, C., D’Haultf\oe uille, X., 2020. Two-way fixed effects estimators with heterogeneous treatment effects. Am. Econ. Rev. 110, 2964–2996. \\
\DOI{10.1257/aer.20181169}.

\item Dobkin, C., Finkelstein, A., Kluender, R., Notowidigdo, M.J., 2018. The economic consequences of hospital admissions. Am. Econ. Rev. 108, 308–352. \\
\DOI{10.1257/aer.20161038}.

\item Fondazione GIMBE, 2025.
Medici di famiglia a rischio estinzione: ne mancano oltre 5.500, il 52\% è sovraccarico di assistiti, 7.300 andranno in pensione entro il 2027.
Press release, 4 March 2025.

\item Fadlon, I. and Nielsen, T.H., 2019. Family health behaviors. 
\emph{American Economic Review}, 109, 3162--3191. \DOI{10.1257/aer.20171993}.

\item Fadlon, I. and Nielsen, T.H., 2021. Family labor supply responses to severe health shocks: Evidence from Danish administrative records.
\emph{American Economic Journal: Applied Economics}, 13, 1--30.

\item Freyaldenhoven, S., Hansen, C.B., Pérez, J.P., Shapiro, J.M., Carreto, C., 2025. xtevent: Estimation and visualization in the linear panel event-study design. Stata J. 25, 97–135. \DOI{10.1177/1536867X251322964}.

\item Freyaldenhoven, S., Hansen, C., Pérez, J.P., Shapiro, J.M., 2021. Visualization, identification, and estimation in the linear panel event-study design. NBER Working Paper No. 29170. National Bureau of Economic Research, Cambridge, MA. \DOI{10.3386/w29170}.

\item Goldsmith-Pinkham, P., Hull, P. and Kolesár, M., 2025. Leniency Designs: An Operator’s Manual (No. w34473). National Bureau of Economic Research.

\item Goodman-Bacon, A., 2021. Difference-in-differences with variation in treatment timing. J. Econometrics 225, 254–277. \DOI{10.1016/j.jeconom.2021.03.014}.

\item Groenewegen, A., Rutten, F.H., Mosterd, A., Hoes, A.W., 2020. Epidemiology of heart failure. Eur. J. Heart Fail. 22, 1342–1356. \DOI{10.1002/ejhf.1858}.

\item Grossman, M., 1972. On the concept of health capital and the demand for health. J. Polit. Econ. 80, 223–255. \DOI{10.1086/259880}.

\item Hoagland, A., 2025. An Ounce of Prevention or a Pound of Cure? The Value of Health Risk Information. 
\emph{Review of Economics and Statistics}, 1--45. \\
\DOI{10.1162/rest.a.276}.

\item Iommi, M., Rosa, S., Fusaroli, M., Rucci, P., Fantini, M.P. and Poluzzi, E., 2020. Modified-Chronic Disease Score (M-CDS): Predicting the individual risk of death using drug prescriptions. \emph{PLoS One}, 15, p.e0240899.

\item Inglis, S.C., Clark, R.A., Dierckx, R., Prieto-Merino, D., Cleland, J.G., 2015. Structured telephone support or non-invasive telemonitoring for patients with heart failure. Cochrane Database Syst. Rev. 10, CD007228. \\
\DOI{10.1002/14651858.CD007228.pub3}.

\item Ito, M., Tajika, A., Toyomoto, R., Imai, H., Sakata, M., Honda, Y., Kishimoto, S., Fukuda, M., Horinouchi, N., Sahker, E. and Furukawa, T.A., 2024. The short and long-term efficacy of nurse-led interventions for improving blood pressure control in people with hypertension in primary care settings: a systematic review and meta-analysis. 
\emph{BMC Primary Care}, 25, p.143.

\item Little, R.J., Rubin, D.B., 2019. Statistical Analysis with Missing Data, third ed. John Wiley \& Sons, Hoboken. \DOI{10.1002/9781119482260}.

\item Lizcano‐Álvarez, Á., Carretero‐Julián, L., Talavera‐Saez, A., Cristóbal‐Zárate, B., Cid‐Expósito, M.G., Alameda‐Cuesta, A., REccAP Group (Red de Enfermería de Cuidados Cardiovasculares en Atención Primaria), Gómez Menor, C., Dionisio Benito, J., Gómez Puente, J.M. and López Köllmer, L., 2023. Intensive nurse‐led follow‐up in primary care to improve self‐management and compliance behaviour after myocardial infarction. Nursing Open, 10(8), pp.5211-5224.

\item Ong, M.K., Romano, P.S., Edgington, S., Aronow, H.U., Auerbach, A.D., Black, J.T., De Marco, T., Escarce, J.J., Evangelista, L.S., Hanna, B., Ganiats, T.G., 2016. Effectiveness of remote patient monitoring after discharge of hospitalized patients with heart failure: The BEAT-HF randomized clinical trial. JAMA Intern. Med. 176, 310–318. \DOI{10.1001/jamainternmed.2015.7712}.

\item Oster, E., 2018. Diabetes and diet: Purchasing behavior change in response to health information.
\emph{American Economic Journal: Applied Economics}, 10, 308--348. \\
\DOI{10.1257/app.20160232}

\item Rambachan, A., Roth, J., 2023. A more credible approach to parallel trends. Rev. Econ. Stud. 90, 2471–2507. \DOI{10.1093/restud/rdad018}.

\item Struttura Interregionale Sanitari Convenzionati (SISAC) (2018). Accordo Collettivo Nazionale per la disciplina dei rapporti con i medici di medicina generale (Triennio 2016–2018). Rep. Atti n. 112/CSR, 21 June 2018.

\item Sun, L., Abraham, S., 2021. Estimating dynamic treatment effects in event studies with heterogeneous treatment effects. J. Econometrics 225, 175–199. \\
\DOI{10.1016/j.jeconom.2020.09.006}.

\item Takeda, A., Taylor, S.J.C., Taylor, R.S., Khan, F., Krum, H., Underwood, M., Richards, A.M., 2012. Clinical service organisation for heart failure. Cochrane Database Syst. Rev. 12, CD002752. \DOI{10.1002/14651858.CD002752.pub3}.

\item Taubman, S.L., Allen, H.L., Wright, B.J., Baicker, K. and Finkelstein, A.N., 2014. Medicaid increases emergency-department use: evidence from Oregon's Health Insurance Experiment. 
\emph{Science}, 343, 263--268. \DOI{10.1126/science.1246183}

\item Wang, Q., Shen, Y., Chen, Y. and Li, X., 2019. Impacts of nurse-led clinic and nurse-led prescription on hemoglobin A1c control in type 2 diabetes: a meta-analysis.
\emph{Medicine}, 98, p.e15971.

\item Wu, X., Li, Z., Tian, Q., Ji, S. and Zhang, C., 2024. Effectiveness of nurse-led heart failure clinic: A systematic review. 
\emph{International Journal of Nursing Sciences}, 11, 315--329.

\end{description}

\end{document}